%% file: main.tex
\documentclass[reprint, superscriptaddress, aps, twocolumn, prx, floatfix]{revtex4-2}
\input{header.tex}

\begin{document}

\title{Structure-Aware Variance Reduction for\\ Unbiased Randomized Hamiltonian Simulation}

\author{Joshua W. Dai} % Add your names and affiliations here!
\email{joshua.dai@materials.ox.ac.uk}
\affiliation{Department of Materials, University of Oxford, Parks Road, OX2 6UD, United Kingdom}
\affiliation{Mathematical Institute, University of Oxford, Woodstock Road, Oxford OX2 6GG, United Kingdom}
\author{Fredrik Hasselgren}
\email{fredrik.hasselgren@maths.ox.ac.uk}
\affiliation{Mathematical Institute, University of Oxford, Woodstock Road, Oxford OX2 6GG, United Kingdom}
\author{Chusei Kiumi}
\email{c.kiumi.qiqb@osaka-u.ac.jp}
\affiliation{Mathematical Institute, University of Oxford, Woodstock Road, Oxford OX2 6GG, United Kingdom}
\affiliation{Center for Quantum Information and Quantum Biology,
The University of Osaka,
1-2 Machikaneyama, Toyonaka, Osaka, Japan}

\date{\today}
\begin{abstract}
Randomized Hamiltonian simulation methods are often governed by a trade-off between systematic bias and sampling overhead.
We study how classical variance-reduction techniques can be applied to such methods without changing their mean channel, and therefore without introducing additional bias.
As a motivating unbiased estimator, we formulate continuous time-evolution probabilistic angle interpolation (continuous TE-PAI), a quasiprobabilistic random-circuit protocol whose remaining Monte Carlo error is purely statistical. Continuous TE-PAI removes Trotter discretization error with finite-depth random circuits, whereas deterministic Trotterization does so only in the infinite-depth limit. Further, in tensor-network simulations, we demonstrate that discretization error can cause an unphysical exponential growth in the bond dimension required for Trotterized simulations, whereas comparable-depth continuous TE-PAI circuits avoid this growth.
We then show that the variance of randomized product-formula-based estimators admits a canonical decomposition into a classical counting component and a quantum ordering component such that the dominant simulation overhead results from the non-commutative parts of the Hamiltonian dynamics.
Motivated by this decomposition, we achieve an $\approx70\%$ error-reduction using the counting-component for small systems whereas our tensor-network simulations of $n=30$ spin-chain dynamics use coarser statistics tailored to the observable and estimator attaining a negligible bias and a reduction of $\approx 80\%$ leading to $\approx91\%$ and $\approx96\%$ sampling-cost reductions, respectively.
\end{abstract}
\maketitle

\section{Introduction}

Hamiltonian simulation is one of the central applications of quantum computers, originating from Feynman's proposal to simulate quantum physics with quantum devices~\cite{Feynman} and later formalized in digital quantum simulation~\cite{Lloyd}.
A wide range of algorithms has since been developed.
Linear-combination-of-unitaries methods~\cite{LCU1,LCU2,LCU3} and quantum-signal-processing or qubitization approaches~\cite{QSP1,QSP2} achieve highly favorable asymptotic scaling, but their circuit constructions can be demanding for near-term and early fault-tolerant devices.
Product-formula methods use simpler circuit primitives and are often more natural in practice, but they generally require long circuits to suppress systematic discretization error.
This deterministic product-formula error is controlled by commutators and nested commutators of Hamiltonian terms, linking simulation cost to the noncommutative structure of the target dynamics.

Randomized Hamiltonian simulation methods offer a complementary approach.
Instead of implementing a single deterministic approximation, they sample random circuit trajectories and estimate observables by averaging over many realizations.
Examples include randomized product formulas~\cite{childs_faster_2019}, qDRIFT and its variants~\cite{qDRIFT1,qDRIFT2,qDRIFT3}, continuous or time-dependent qDRIFT-type methods~\cite{berry_time-dependent_2020}, randomized multiproduct formulas~\cite{faehrmann_randomizing_2022}, importance-sampled stochastic simulation~\cite{kiss_importance_2023}, and related random compilers.
Randomization can reduce circuit depth or simplify sampling, but it introduces statistical error.
Thus randomized simulation algorithms are naturally described by a bias--variance decomposition: the mean sampled channel may differ from the target channel, while finite-sample estimators fluctuate around that mean.

This paper takes this bias--variance viewpoint as its organizing principle.
For biased randomized algorithms such as qDRIFT, variance reduction can reduce sampling cost but cannot remove the mean-channel bias.
For unbiased randomized algorithms, by contrast, variance is the sole intrinsic source of Monte Carlo error.
This makes unbiased trajectory estimators a natural setting for structure-aware variance reduction.

As a motivating unbiased estimator, we formulate continuous TE-PAI for time-dependent Pauli Hamiltonians.
The protocol builds on probabilistic angle interpolation~\cite{PAI} and finite-step TE-PAI~\cite{Kiumi_2025}, but removes the remaining product-formula discretization error by passing to continuous time. Crucially, finite-step TE-PAI removes this error only as \(N\to\infty\), where its classical sampling cost \(O(NLN_s)\) diverges, whereas the continuous protocol achieves zero discretization error with finite classical-quantum resources per sample. It samples a marked Poisson process of fixed-angle Pauli rotations and assigns a scalar quasiprobability weight to each sampled physical circuit.
The resulting weighted random circuit is an unbiased estimator of the exact time-evolution channel.
A continuous-time PAI protocol based on the same marked-Poisson sampling principle was independently introduced in concurrent work by Hayata and Kikuchi~\cite{Hayata_2026}.
An additional contribution of our work is a formal derivation of the continuous protocol and a rigorous proof of its unbiasedness. Here we give two derivations of unbiasedness: one as the weak limit of finite-step TE-PAI, and one directly from the Dyson-series expansion of the exact time-evolution channel~\cite{Dyson_1949}.
Beyond establishing unbiasedness, the two derivations clarify complementary aspects of the protocol: the weak-limit argument shows that the continuous-time random-circuit ensemble arises as the limiting law of finite-step TE-PAI, while the direct Dyson-series construction identifies the same ensemble as a Poisson sampling of infinitesimal time-evolution contributions.

Unbiased randomized Hamiltonian simulation has also been studied through Dyson-series-based random compilers and stochastic observable-dynamics algorithms~\cite{Zhang2022,granet_hamiltonian_2024}.
These approaches avoid product-formula discretization error by averaging over random circuit contributions with known classical prefactors, but their measurement procedures involve Hadamard-test-like controlled measurements using ancilla qubits.
By contrast, continuous TE-PAI offers the practical advantage of sampling ordinary ancilla-free Pauli-rotation circuits, requiring only a scalar quasiprobability weight in classical postprocessing~\cite{Kiumi_2025, Toshio_2026}.

These advantages are particularly relevant in tensor-network simulation, where the cost of contracting each sampled trajectory depends strongly on the entanglement it generates. 
Recent work has adapted TE-PAI to matrix-product-state simulation, emphasizing its embarrassingly parallel structure and its ability to trade a single deep deterministic Trotter circuit for an ensemble of shallower random circuits~\cite{Fredrik_2026}. We show numerically that continuous TE-PAI can avoid the discretization-induced bond-dimension growth observed in shallow Trotterized simulations.

Once the systematic bias is removed, the central remaining cost is statistical.
For continuous TE-PAI this cost is controlled in the worst case by a quasiprobability overhead depending on the time-integrated \(\ell_1\) norm of the Hamiltonian coefficients.
The TE-PAI angle \(\Delta\) trades expected circuit depth against sampling overhead.
However, this worst-case bound treats the Hamiltonian essentially through absolute coefficient weights and does not exploit commutation relations, locality, observable support, or other structure. The natural question is how do we use this additional knowledge to reduce the practical variance of randomized simulation.

We address this question using conditional sampling and stratification, also known as Rao--Blackwellization.
Recent work of Dai and Koczor~\cite{Dai_2025} developed a stratified-sampling framework for quasiprobability decompositions, proving that ideal proportional allocation gives an unbiased estimator with variance no worse than naive sampling and demonstrating constant-factor reductions.
Multilevel Monte Carlo methods have also recently been applied to qDRIFT by coupling estimators at different circuit depths~\cite{mohammadipour2026mlmc}.
Our contribution builds upon~\cite{Dai_2025} in the specific context of Hamiltonian simulation.

The key object is the counts vector, which maps a trajectory to the number of times each alphabet letter appears.
Conditioning on the counts vector decomposes trajectory variance into a classical counting component and a quantum ordering component.
The counting component reflects fluctuations in which elementary gates are sampled.
The ordering component is the residual variance within a fixed sampled multiset and is therefore caused only by different orderings of the same gates. This is illustrated in \cref{fig:variance graph}.
In exchangeable models such as qDRIFT, each count sector is precisely a random-permutation product-formula ensemble. After conditioning away classical count fluctuations, the remaining ordering variance is governed by noncommutativity, in direct analogy with the commutator structure of deterministic product-formula error.

The full counts vector is not generally a scalable stratification statistic.
For trajectories of length \(m\) over an alphabet of size \(q\), its support size is \(\binom{m+q-1}{q-1}\), which becomes prohibitive for large Hamiltonian alphabets.
We therefore use it primarily as a canonical diagnostic and organizing principle.
It motivates lower-dimensional statistics tailored to the estimator and observable, including local counts, grouped counts, time-binned counts, and sign-aware statistics for quasiprobability estimators.
For TE-PAI, sign-aware stratification is important because the sampled circuit is physical but the estimator carries a quasiprobability sign determined by the parity of \(\pi\)-rotation events.

Our numerical experiments test these ideas in continuous TE-PAI simulations of spin-chain dynamics. For an \(n=8\) transverse-field Ising chain and a \(12\)-qubit spin-chain example, observable-adapted local-count and sign-aware \(\pi\)-rotation stratification reduce the trajectory-sampling error by roughly \(70\%\), corresponding to an approximately \(91\%\) reduction in the required number of trajectory samples. In tensor-network simulations of \(n=30\) time-dependent spin-chain dynamics, coarser observable-adapted statistics reduce the trajectory-sampling error by about \(80\%\), corresponding to an approximately \(96\%\) sampling cost reduction. These sample-count estimates refer to trajectory sampling and do not include measurement shot noise. Also, the same simulations show that continuous TE-PAI avoids the discretization-induced bond-dimension growth observed in shallow Trotterized simulations. 

Our results with simple stratification suggest that these constant-factor reductions could offer practical advantages in applications. Moreover, the optimal stratified-sampling scheme is problem dependent and can be further improved by tailoring the strata to the Hamiltonian, observable, and estimator. Our work provides a foundation for such problem-specific designs and for further reductions in sampling overhead.

The paper is organized as follows.
In \cref{sec:randomized_trajectory_estimators}, we introduce a general random-trajectory framework and formulate the bias--variance decomposition at the channel and observable levels.
In \cref{sec:continuous_te_pai}, we define continuous TE-PAI, prove its unbiasedness both as a weak limit of finite-step TE-PAI and through a Dyson-series construction, and discuss its tensor-network implications.
In \cref{sec:structure_aware_variance_reduction}, we develop the structure-aware variance theory, including counts-vector abelianization and Poincar\'e-type bounds in terms of swap and substitution defects, and illustrate the decomposition using qDRIFT as a biased randomized trajectory estimator.
Finally, in \cref{sec:numerics}, we specialize these ideas to continuous
TE-PAI and benchmark sign-aware and observable-adapted stratification
strategies in spin-chain simulations, including tensor-network simulations at
larger system sizes.

\section{Randomized trajectory estimators for Hamiltonian simulation}
\label{sec:randomized_trajectory_estimators}

We consider a finite-dimensional system with Hilbert space \(\mathcal H\), \(\dim\mathcal H=d\), evolving under a time-dependent Pauli Hamiltonian
\[
    H(t)=\sum_{k=1}^{L}c_k(t)P_k ,
\]
where \(c_k(t)\in\mathbb R\) are absolutely continuous on \([0,T]\). 
The target unitary and channel are
\begin{equation}
    U(T)
    =
    \mathcal T\exp\left[-i\int_0^T H(t)\,dt\right],
    \,\, 
    \mathcal U_T(\rho)=U(T)\rho U(T)^\dagger .
    \label{eq:time-evolution-operator}
\end{equation}

Randomized simulation algorithms replace \(\mathcal U_T\) by a random superoperator \(\widehat\Phi\) taking values in
\[
    \mathsf S:=\mathcal L(\mathcal L(\mathcal H)),
\]
equipped with the Hilbert--Schmidt norm.  Observable-estimation tasks are fixed linear functionals \(L:\mathsf S\to\mathbb C\); for instance
\[
    L_{\rho,O}(\Phi)=\Tr[O\,\Phi(\rho)] .
\]
If \(\overline\Phi=\mathbb E[\widehat\Phi]\), then the Monte Carlo estimator
\[
    \widehat\mu_{N_s}
    =
    \frac1{N_s}\sum_{r=1}^{N_s}L(\widehat\Phi_r)
\]
obeys the bias--variance decomposition
\begin{equation}
    \mathbb E\left|
    \widehat\mu_{N_s}-L(\mathcal U_T)
    \right|^2
    =
    \left|L(\overline\Phi-\mathcal U_T)\right|^2
    +
    \frac1{N_s}\Var\!\left(L(\widehat\Phi)\right).
    \label{eq:observable_bias_variance}
\end{equation}
Thus deterministic product formulas have zero sampling variance but generally nonzero bias, while randomized algorithms may have both.
Once an estimator is unbiased at the channel level, \(\overline\Phi=\mathcal U_T\), the only remaining algorithmic error in \cref{eq:observable_bias_variance} is statistical.

We will use the channel-level risk
\[
    \mathcal R(\widehat\Phi)
    :=
    \mathbb E\left\|
        \widehat\Phi-\mathbb E[\widehat\Phi]
    \right\|_{HS}^{2},
\]
which controls every fixed observable functional through
\[
    \Var(L(\widehat\Phi))
    \le
    \|L\|^2\mathcal R(\widehat\Phi),
    \qquad
    \|L\|
    :=
    \sup_{\Psi\neq0}
    \frac{|L(\Psi)|}{\|\Psi\|_{HS}}.
\]
This is the quantity we aim to reduce.

\subsection{Trajectory estimators}

Many randomized simulation algorithms can be viewed as random trajectories of elementary channels. 
Let \(\mathcal A\) be a finite alphabet and let
\[
    \mathcal W:=\bigsqcup_{m\ge0}\mathcal A^m
\]
be the set of finite words. 
A word \(\underline x=(x_1,\ldots,x_m)\) is realized as the channel
\[
    \Gamma(\underline x)
    :=
    \Gamma_{x_m}\circ\cdots\circ\Gamma_{x_1},
    \qquad
    \Gamma(\varnothing):=\mathrm{id},
\]
where each letter \(a\in\mathcal A\) is assigned an implementable unitary
channel \(\Gamma_a\). 
A random word \(\underline X\) therefore gives the trajectory estimator
\[
    \widehat\Phi=\Gamma(\underline X).
\]
This framework includes randomized product formulas and qDRIFT-type algorithms, where the letters label Hamiltonian terms.

Quasiprobability estimators have the same sampled physical circuits but include a scalar weight (this is an importance sampling scheme)
\[
    \widehat\Phi
    =
    G(\underline X)\Gamma(\underline X),
    \qquad
    G:\mathcal W\to\mathbb R.
\]
The weighted object is generally not a quantum channel when \(G<0\) or \(G\ne1\), but it is a well-defined element of \(\mathsf S\) whenever \(\mathbb E\|\widehat\Phi\|_{HS}^2<\infty\).  Continuous TE-PAI, introduced in the next section, is of this form: \(G\) has fixed magnitude, with a sign determined by the parity of \(\pi\)-rotation events. Although continuous TE-PAI trajectories also carry event times, the same
conditional-expectation identities hold on the full marked point-process space. 
In applications we condition only on finite statistics of those trajectories, such as Pauli labels, local counts, time bins, or sign parities.

\subsection{Stratification and conditional risk}

Let \(S=S(\underline X)\) be a trajectory statistic. 
Conditioning on \(S\) partitions the trajectory space into strata, and the Hilbert-space law of total variance gives
\begin{equation}
    \mathcal R(\widehat\Phi)
    =
    \underbrace{
    \mathbb E\!\left[\Var(\widehat\Phi\mid S)\right]
    }_{\mathcal R(\widehat\Phi,S)}
    +
    \underbrace{
    \mathbb E\left\|
        \mathbb E[\widehat\Phi\mid S]
        -
        \mathbb E[\widehat\Phi]
    \right\|_{HS}^{2}
    }_{\text{variance explained by }S}.
    \label{eq:conditional_risk}
\end{equation}
The first term,
\[
    \mathcal R(\widehat\Phi,S)
    :=
    \mathbb E\!\left[\Var(\widehat\Phi\mid S)\right],
\]
is the residual risk after stratifying by \(S\).  The Rao--Blackwellized estimator \(\mathbb E[\widehat\Phi\mid S]\) has the same mean as \(\widehat\Phi\) and no larger variance. 
In practice this conditional expectation is rarely computed exactly, but \cref{eq:conditional_risk} identifies which trajectory information is responsible for variance.

For continuous TE-PAI, the estimator is unbiased:
\[
    \mathbb E_\omega[g_\omega\mathcal U_\omega]=\mathcal U_T.
\]
Therefore
\[
    \mathbb E\left|
    \widehat\mu_{N_s}-L(\mathcal U_T)
    \right|^2
    =
    \frac1{N_s}
    \Var\!\left(L(g_\omega\mathcal U_\omega)\right).
\]
After establishing this unbiasedness, the remaining problem is to choose statistics \(S\) that reduce the variance of the weighted trajectory estimator without changing its mean.
\section{Continuous TE-PAI: an unbiased trajectory estimator}
\label{sec:continuous_te_pai}
\subsection{Continuous TE-PAI}
We now define a continuous random-circuit estimator for the channel
$\mathcal U_T$. Fix a TE-PAI parameter $\Delta\in(0,\pi)$. Define the total rate
\begin{equation}
\Lambda
:=
G_\Delta\,
\overline{\|c\|_1}\,T ,\quad G_\Delta:=\frac{3-\cos\Delta}{\sin\Delta},
\label{eq:Lambda}
\end{equation}
where $\overline{\|c\|_1}:=\frac{1}{T}\int_0^T \sum_{k=1}^L |c_k(t)|\,dt$. The sample space is $\Omega=\bigsqcup_{M\geq 0}\Omega_M$, where each $\Omega_M$ consists of ordered sequences of event triples $(t_m,k_m,\ell_m)$ with times $t_m\in[0,T]$, Pauli indices $k_m\in[L]:=\{1,\ldots,L\}$, and type labels $\ell_m\in\{0,1\}$.

\medskip\noindent\textbf{Algorithm: Continuous TE-PAI.}
\begin{enumerate}[leftmargin=*,itemsep=2pt]
\item \textbf{Sample gate count.} Draw $M\sim\mathrm{Poisson}(\Lambda)$.
\item \textbf{Sample times.} Draw $M$ i.i.d.\ times from $[0,T]$ with density
$f(t)={\|c(t)\|_1}/({T\,\overline{\|c\|_1}})$
and sort them as $0\leq t_1\leq\cdots\leq t_M\leq T$.
\item \textbf{Sample Pauli indices.} For each $t_m$, draw $k_m\in[L]$ with $\Pr(k_m=k\mid t_m)={|c_k(t_m)|}/{\|c(t_m)\|_1}$.
\item \textbf{Sample angles.} For each $m$, draw $\ell_m\in\{0,1\}$ with
\[
p_\Delta:=\Pr(\ell_m=0)=\frac{2}{3-\cos\Delta},\quad p_\pi:=1-p_\Delta.
\]
\item \textbf{Build circuit and weight.} Construct the random unitary $U_\omega=\prod_{m=1}^M R_m$ where
\[
R_m=\begin{cases}
R_{P_{k_m}}(\mathrm{sgn}(c_{k_m}(t_m))\,\Delta) & \ell_m=0,\\
R_{P_{k_m}}(\pi) & \ell_m=1.
\end{cases}
\]
Assign the total weight
\begin{equation}
g_\omega =
\left(\prod_{m=1}^M(-1)^{\ell_m}\right)
\exp\!\Big[2\,\overline{\|c\|_1}\,T\tan\!\Big(\frac{\Delta}{2}\Big)\Big].
\label{eq:weight}
\end{equation}
\end{enumerate}
Equivalently, with respect to Lebesgue measure on the ordered simplex
\(0\le t_1\le\cdots\le t_M\le T\) and counting measure over the discrete
labels, the trajectory density of
\(\omega=(t_m,k_m,\ell_m)_{m=1}^M\in\Omega_M\) is
\[
P_\infty(\omega)
=
e^{-\Lambda}\frac{\Lambda^M}{M!}\,M!
\prod_{m=1}^{M}
\left[
\frac{|c_{k_m}(t_m)|}{T\overline{\|c\|_1}}\,
p_\Delta^{1-\ell_m}p_\pi^{\ell_m}
\right].
\]
Here \(e^{-\Lambda}\Lambda^M/M!\) gives the Poisson event-count law, and the
factor \(M!\) accounts for sorting the event times. The signs from
\(\pi\)-rotation events are not part of \(P_\infty\), but are included in the
quasiprobability weight \(g_\omega\). Thus \(P_\infty\) samples physical
circuits, while \(g_\omega\) supplies the signed weight. The resulting weighted
circuit estimator is unbiased, as stated below.
\begin{theorem}\label{thm:unbiased}
Let $\mathcal U_T$ be the channel of the exact time evolution $U(T)$, and let
$\mathcal U_\omega$ be the channel of the random circuit $U_\omega$. Then
\[
\mathbb E_{\omega}\left[g_\omega\,\mathcal U_\omega\right]=\mathcal U_T .
\]
Moreover, for any observable \(O\) with \(\|O\|\le 1\), a single-shot outcome
\(Z_\omega\in[-1,1]\) gives an unbiased estimator
\(Y_\omega:=g_\omega Z_\omega\) of
\(\langle O\rangle=\Tr[O\,\mathcal U_T(\rho_0)]\). To estimate it to additive
error \(\varepsilon\), it suffices to use
\[
N_s
=
\mathcal O\!\left(
\exp\!\bigl[4\,\overline{\|c\|_1}\,T\tan(\Delta/2)\bigr]
/\varepsilon^2
\right)
\]
independent single-shot circuit samples. The corresponding resource costs are:
\begin{enumerate}[itemsep=0pt]
  \item Expected gate count per circuit:
        \[
          \mathbb{E}[M]
          =
          \Lambda
          =
          G_\Delta
          \overline{\|c\|_1}T .
        \]
  \item Classical preprocessing for trajectory sampling:
        \[
          \mathcal{O}\!\left(
          \Lambda\,
          \exp\!\bigl[4\,\overline{\|c\|_1}\,T\tan(\Delta/2)\bigr]
          /\varepsilon^{2}
          \right).
        \]
  \item Total circuit executions:
        \[
          \mathcal{O}\!\left(
          \exp\!\bigl[4\,\overline{\|c\|_1}\,T\tan(\Delta/2)\bigr]
          /\varepsilon^{2}
          \right).
        \]
\end{enumerate}
\end{theorem}

We provide two complementary derivations of \cref{thm:unbiased} in \cref{sec:derivation1,sec:dyson}: one from the
weak-limit convergence of finite-step TE-PAI~\cite{Kiumi_2025}, which connects
the continuous protocol to the well-studied finite-angle PAI framework, and one
from a Dyson-series representation of the time-evolution channel, which gives a
direct physical interpretation as Poisson sampling of infinitesimal
time-evolution contributions.

\paragraph*{Remark: Trade-off in the TE-PAI angle.}
The parameter \(\Delta\in(0,\pi)\) controls the trade-off between circuit depth
and sampling overhead. The overhead increases monotonically with \(\Delta\),
whereas \(G_\Delta\) is minimized at \(\Delta=\arccos(1/3)\approx 1.23\);
hence the useful regime is \(0<\Delta\le \arccos(1/3)\). At the small-angle
end, choosing \(\Delta=\Theta((\overline{\|c\|_1}T)^{-1})\) keeps the sampling
overhead constant, since
\(\exp[4\overline{\|c\|_1}T\tan(\Delta/2)]=O(1)\), while
\(\mathbb{E}[M]=G_\Delta\overline{\|c\|_1}T
=O((\overline{\|c\|_1}T)^2)\). Thus continuous TE-PAI can achieve constant
sampling overhead while remaining unbiased, at the cost of Trotter-like
quadratic circuit-depth scaling.

\paragraph*{Remark: Term-dependent TE-PAI angles.}
The TE-PAI angle may be chosen separately for each Pauli term. Assigning a
larger angle \(\Delta_k\) to costly or low-fidelity rotations reduces their
event rate, while smaller angles can be used for cheaper rotations to control
the sampling overhead. This gives a hardware-aware tradeoff between execution
cost and sampling variance. For example, let \(p_1\) and \(p_2\) be representative one- and two-qubit gate
error rates. Since two-qubit gates are often the dominant error source, with
ratios such as \(p_2/p_1\sim 10\) on current devices
\cite{mckay2023benchmarking, google_willow_2024}, it is natural to assign
larger angles to rotations whose compiled circuits contain two-qubit gates.
Let \(K_h\) be this expensive subset and \(K_s\) the remaining cheaper subset,
with angles \(\Delta_h\) and \(\Delta_s\), respectively. Since the TE-PAI
event-rate factor is \(G_{\Delta}\), reducing the event rate of the expensive
rotations in proportion to the gate-error ratio amounts to choosing
$
G_{\Delta_h}/G_{\Delta_s}\approx p_1/p_2.
$
In the small-angle regime, \(G_{\Delta}\approx 2/\Delta\), and hence
\(\Delta_h\approx (p_2/p_1)\Delta_s.\)
Thus, for \(p_2/p_1\approx 10\), one may choose \(\Delta_h\approx 10\Delta_s\),
making two-qubit-heavy rotations occur about ten times less often. The price is increased sampling overhead. If
\[
C_h:=\sum_{k\in K_h}\int_0^T |c_k(t)|\,dt ,
\]
then replacing \(\Delta_s\) by \(\Delta_h\) on \(K_h\) multiplies the sample
complexity by
\[
\exp\left\{
4C_h\left[
\tan(\Delta_h/2)-\tan(\Delta_s/2)
\right]
\right\}.
\]
For small angles and \(\Delta_h\approx (p_2/p_1)\Delta_s\), this becomes
\[
\exp\left[2C_h\left(\frac{p_2}{p_1}-1\right)\Delta_s\right].
\]
Hence term-dependent TE-PAI angles trade a linear reduction in the event rate
of costly rotations for an exponential sampling-cost increase controlled by
their total weight \(C_h\).

\subsection{Derivation from finite-step TE-PAI}
\label{sec:derivation1}
We first obtain the protocol as the continuous-time limit of finite-step TE-PAI. Finite-step TE-PAI~\cite{Kiumi_2025} applies the PAI quasiprobability
decomposition~\cite{PAI} to a Trotterized evolution, sampling fixed-angle Pauli-rotation
circuits whose reweighted mean reproduces the Trotter channel exactly; such fixed-angle
rotations are well suited to early fault-tolerant hardware~\cite{Toshio_2026}, and the
same randomized strategy extends to tasks such as energy-gap estimation~\cite{Pages_2026}.
The discretization error vanishes only in the limit \(N\to\infty\), but the
classical cost of generating finite-step samples scales as \(O(NL)\) per sample
and therefore diverges in this limit. This motivates the continuous-time
sampler.
Partition \([0,T]\) into \(N\) intervals of width \(\delta t=T/N\), with \(t_j=j\delta t\), and define
\begin{equation}
    \mathcal U_N
:=
\prod_{j=1}^{N}\prod_{k=1}^{L}
\mathcal R_{P_k}(\theta_{kj}),
\quad
\theta_{kj}:=2c_k(t_j)\delta t ,
\label{eq:trotter-circuit}
\end{equation}
where \(\mathcal R_{P_k}(\theta)(\rho):=R_{P_k}(\theta)\rho R_{P_k}(\theta)^\dagger\) and \(R_{P_k}(\theta):=\exp(-i\theta P_k/2)\). For each elementary channel, finite-step TE-PAI uses the local PAI
decomposition
\[
\mathcal R_P(\theta)
=
a_1(\theta)\mathcal R_{P,1}
+
a_2(\theta)\mathcal R_{P,2}
+
a_3(\theta)\mathcal R_{P,3},
\]
where \(\mathcal R_{P,1}\), \(\mathcal R_{P,2}\), and
\(\mathcal R_{P,3}\) are the channels corresponding to \(I\),
\(R_P(\mathrm{sgn}(\theta)\Delta)\), and \(R_P(\pi)\), respectively. Let
\[
g(\theta)=\sum_{r=1}^{3}|a_r(\theta)|,
\qquad
p_r(\theta)=\frac{|a_r(\theta)|}{g(\theta)} .
\]
For each pair \((j,k)\), sample \(r_{kj}\in\{1,2,3\}\) with probability
\(p_{r_{kj}}(\theta_{kj})\). Identity outcomes \(r_{kj}=1\) are not
recorded. If \(r_{kj}=2\), record the event \((t_j,k,0)\); if
\(r_{kj}=3\), record the event \((t_j,k,1)\). The recorded events form a
trajectory \(\omega=(t_m,k_m,\ell_m)_{m=1}^M\), and the corresponding
circuit is obtained by applying, in time order,
\[
R_m=
\begin{cases}
R_{P_{k_m}}(\mathrm{sgn}(c_{k_m}(t_m))\Delta), & \ell_m=0,\\
R_{P_{k_m}}(\pi), & \ell_m=1.
\end{cases}
\]
The trajectory weight is
\[
g_\omega^{(N)}
=
\prod_{j=1}^{N}\prod_{k=1}^{L}
\operatorname{sgn}\!\left(a_{r_{kj}}(\theta_{kj})\right)
g(\theta_{kj}).
\]
Let \(I=[N]\times[L]\). For a trajectory \(\omega\), let
\(I_1^{(\omega)}\subset I\) be the set of indices corresponding to identity
outcomes. Similarly, let \(I_2^{(\omega)}\) and \(I_3^{(\omega)}\) be the sets
of indices \((j,k)\in I\) for which \((t_j,k,0)\in\omega\) and
\((t_j,k,1)\in\omega\), respectively. Then the probability of \(\omega\) is
\[
P_N(\omega)
=
\prod_{r=1}^{3}
\prod_{(j,k)\in I_r^{(\omega)}}
p_r(\theta_{kj}) .
\]

Let \(\mathcal U_\omega^{(N)}\) denote the channel induced by this finite-step
random circuit. By construction, finite-step TE-PAI is unbiased for the product-formula
channel~\cite{Kiumi_2025}:
\[
\mathcal U_N
=
\mathbb E_{P_N}\!\left[
g_\omega^{(N)}\mathcal U_\omega^{(N)}
\right].
\]

Let \(N_s\) be the number of circuit variants. The classical cost of generating
finite-step TE-PAI circuits is \(O(NLN_s)\), which diverges in the
zero-discretization-error limit \(N\to\infty\). This motivates the
continuous-time sampler. The following lemma makes this finite-to-continuous connection precise. Its
detailed proof is given in \cref{app:weak:limit}.

\begin{lemma}\label{lemma:weak_limit}
\(P_N\) converges weakly to \(P_\infty\) on the trajectory space \(\Omega\);
that is, for every bounded continuous functional \(F:\Omega\to\mathbb C\),
\[
    \lim_{N\to\infty} \int_{\Omega} F(\omega)\,P_N(d\omega)
    =
    \int_{\Omega} F(\omega)\,P_\infty(d\omega).
\]
\end{lemma}
This lemma provides the probabilistic bridge between finite-step and continuous
TE-PAI: as \(N\to\infty\), the discrete random trajectories converge in
distribution to the marked Poisson trajectories \(P_\infty\). Physically, the
rare nontrivial gates on the fine Trotter grid become independent
continuous-time Poisson events, with rates governed by \(\|c(t)\|_1\). Thus
the result identifies the limiting random-circuit ensemble underlying
continuous TE-PAI, rather than only the limiting mean channel. 

\paragraph*{Remarks: Local optimality}
The weak-limit construction also transfers local properties of the finite-step
PAI decomposition to the continuous sampler. For continuous TE-PAI, assume the available gates are
fixed-angle Pauli rotations containing \(\mathcal I\) and
\(\mathcal R_P(\pm\Delta)\), where \(\Delta\in(0,\pi)\) is the smallest
nonzero available angle. Applying the finite-angle PAI optimality
result~\cite{PAI} locally and taking the continuous-time limit, the optimal
local signed decomposition uses \(\mathcal I\), the nearest nonzero rotations
\(\mathcal R_P(\pm\Delta)\), and the antipodal rotation
\(\mathcal R_P(\pi)\). Since the finite-step gamma factors converge,
\(\lim_{N\to\infty}\gamma_N=\gamma\), the continuous TE-PAI gamma overhead is
optimal under the same setting.

\paragraph*{Remarks: Alternative gate sets.}
The gate set used here is not unique. The choice
\(\{I,R_P(\pm\Delta),R_P(\pi)\}\) is natural for Pauli-generated channels
because it exploits \(P^2=I\), but other local decompositions may be better for
structured Hamiltonians. For example, one could use physical channels such as
fermionic hopping rotations, density-density phase gates, or small
Hubbard-block evolutions for Hubbard-type models. Optimizing this local
decomposition is an additional degree of freedom that may reduce the sampling
overhead.

\subsection{Derivation from Dyson-series representation}
\label{sec:dyson}

The continuous TE-PAI sampler can be understood as a probabilistic rewriting of the Dyson expansion; see \cref{app:dyson_derivation_continuous_tepai}. The Dyson series expands the channel in terms of ordered insertions of the nonphysical generators \(\mathcal L_k\). The local TE-PAI identity replaces each such insertion by a signed mixture of physical rotations, and normalizing the positive part of the resulting coefficients produces a marked Poisson trajectory law.

\begin{align*}
\mathcal U(T)
&=
\sum_{M=0}^{\infty}
\sum_{k_1,\ldots,k_M}
\int_{0\le t_1\le\cdots\le t_M\le T}
dt_1\cdots dt_M\\
&\qquad\times
\left(
\prod_{m=1}^M c_{k_m}(t_m)
\right)
\mathcal L_{k_M}\cdots\mathcal L_{k_1}.
\end{align*}

The key local identity behind continuous TE-PAI is the channel-level
decomposition of a Pauli Liouvillian insertion into implementable rotations.
For a Pauli string \(P\), define \(\mathcal L_P(\rho):=-i[P,\rho]\). Then one has
\[
\pm\mathcal L_P
=
\frac{2}{\sin\Delta}
\left[
\mathcal R_P(\pm\Delta)
-
\cos^2\frac{\Delta}{2}\mathcal I
-
\sin^2\frac{\Delta}{2}\mathcal R_P(\pi)
\right].
\]
Exact time evolution accumulates infinitesimal Hamiltonian rotations.
Continuous TE-PAI replaces each infinitesimal Liouvillian insertion, in
expectation, by fixed finite-angle rotations; see \cref{fig:dyson-pai} for a
geometric illustration. The \(\Delta\)-rotation supplies the desired Liouvillian
component, while the identity and \(\pi\)-rotation channels cancel the unwanted
even component of the finite rotation. Equivalently, \(\mathcal R_P(\pi)\) acts
as a sign flip on the \(P\)-anticommuting sector, leaving only the
infinitesimal Liouvillian direction. Applying this local rewriting to each Dyson insertion reorganizes the exact
time-evolution channel as the signed Poisson ensemble
\[
\mathcal U_T
=
\int_{\Omega}
dP_\infty(\omega)\,
g_\omega\,\mathcal U_\omega .
\]
This identity gives a direct derivation of continuous TE-PAI from the Dyson series.

\begin{figure}
    \centering
    \includegraphics[width=0.9\linewidth]{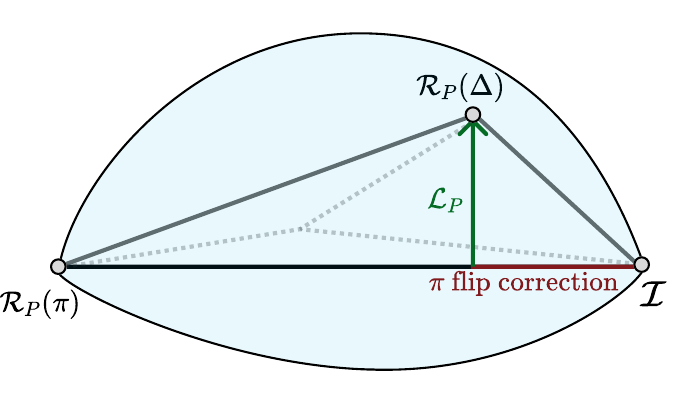}
    \caption{Geometric intuition for the local TE-PAI decomposition.
For a fixed Pauli string \(P\), the finite channels \(\mathcal I\),
\(\mathcal R_P(\Delta)\), and \(\mathcal R_P(\pi)\) give a linear
representation of the infinitesimal direction \(\mathcal L_P\).
The \(\Delta\)-rotation contains both the desired Liouvillian component and an
unwanted component along the \(\mathcal I\)--\(\mathcal R_P(\pi)\) direction.
TE-PAI subtracts this unwanted component, representing the Liouvillian
insertion as a signed combination of physical finite-angle channels.
}
    \label{fig:dyson-pai}
\end{figure}
 
Related Dyson-series random compilers for time-dependent Hamiltonian
simulation have also been proposed, e.g., Ref.~\cite{Zhang2022}. These methods also sample ordered Dyson
terms, but typically estimate unitary contributions to a linear-combination
representation using controlled-unitary measurements. In
contrast, continuous TE-PAI applies the local PAI identity to convert each
nonphysical Dyson insertion into a signed mixture of implementable fixed-angle
Pauli-rotation channels. Thus each sample is an ordinary physical random
circuit, followed by a standard observable measurement and a scalar
quasiprobability weight in classical postprocessing.

\cref{thm:unbiased} shows that continuous TE-PAI removes the systematic bias of product-formula simulation. The remaining cost is statistical: the worst-case sampling overhead is controlled by the \(\ell_1\)-type scale \(\overline{\|c\|_1}T\). This bound treats the Hamiltonian terms only through their absolute weights and does not yet use any additional structure of the problem, such as commutation relations, locality, or the target observable. Thus there is substantial room to reduce the practical variance.

\subsection{Discretization error}
\subsubsection{Algorithmic error}
Trotterization incurs a discretization error from the Lie--Trotter product formula \cref{eq:first-order-trotter}:
\begin{equation}
\label{eq:first-order-trotter}
e^{-i(A+B)t}
=
\lim_{N \to \infty}
\left(
e^{-iA t/N}
e^{-iB t/N}
\right)^N ,
\end{equation}
whenever a finite $N$ is used. Here the time-evolution operator \cref{eq:time-evolution-operator} is approximated by a series of $\exp(-iH(t) \frac{T}{N})$ operators implemented via the circuit $\mathcal{U}_N$ using rotation-gates of angle $\theta_{kj}=2c_k(t_j)\frac{T}{N}$ for Hamiltonian coefficient $k$ and time-step $j$ from \cref{eq:trotter-circuit}. The discretization error leads to a time-evolution of an approximate Hamiltonian \cite{gentinetta_correcting_2025}, using rotation angles that scale inversely with the number of Trotter steps. As a result the additive approximation error of first-order Trotterization is bounded by \cite{childs_theory_2021}:
\begin{align*}
\varepsilon_{\mathrm{Trot}}^{(N)}
&:=
\left\|
\prod_{j=1}^{N}\prod_{k=1}^{L}
e^{-i c_k(t_j)P_k T/N}
-
U(T)
\right\|
\nonumber\le
\frac{T^{2}}{2N}\,\|c\|_{T}^{2}.
\end{align*}
Here, the averaged commutator scale is defined as
\begin{align*}
\|c\|_{T}^{2}
:=
\frac{1}{N}
\sum_{j=1}^{N}
\sum_{\gamma_1=1}^{L}
\left\|
\left[
\sum_{\gamma_2>\gamma_1}
c_{\gamma_2}(t_j)P_{\gamma_2},
\,
c_{\gamma_1}(t_j)P_{\gamma_1}
\right]
\right\|.
\end{align*}
For simplicity, we state the bound for the time-independent case; in the time-dependent setting, additional discretization errors arise from approximating the time-ordered evolution by a finite sequence of time slices, and sharper bounds include terms involving time derivatives of the Hamiltonian~\cite{mizuta_explicit_2024}. Continuous TE-PAI entirely removes this error by replicating the behaviour of Trotter circuits that would require an infinite number of gates to implement. As such, the discretization error in estimating time-evolved observables is replaced by statistical sampling error arising from the intrinsic variance of the continuous TE-PAI estimator when evaluated with a finite number of samples.

\subsubsection{Tensor-network error}
In the case of tensor-network simulation there is an additional avenue of discretization error arising from the erroneous entanglement structure governed by the approximate Hamiltonian. Generally speaking the number of Schmidt coefficients, called the bond-dimension $\chi$, needed across a given cut in an exact tensor network grows exponentially with the entanglement across that cut. Therefore, consider the discrepancy between the target state $\ket{\psi(T)}$ and the simulated state $\ket{\tilde{\psi}(T)}$:
\begin{align*}
    \ket{\psi(T)}
    &=
    U(T)\ket{\psi_0}, \\
    \ket{\tilde{\psi}(T)}
    &=
    \mathcal U_N\ket{\psi_0}
    =
    \ket{\psi(T)}+\epsilon_{\text{Trot}}^{(N)}\ket{\eta},
\end{align*}
where $\ket{\eta}$ is a normalized error vector. Across a cut in the tensor-network, this erroneous vector $\ket{\eta}$ can have support over additional Schmidt sectors and thus introduce a tail of Schmidt coefficients independent of the dynamics of the target state. These values can contribute to growth in the bond-dimension required to simulate $\ket{\tilde \psi}$, through introducing additional interactions and entanglement in the Trotterized Floquet Hamiltonian \cite{pastori_characterization_2022}. Hence the discretization error of the Trotter protocol can introduce additional erroneous bond-dimension growth that does not reflect the entanglement build-up of the underlying system \cite{paeckel_time-evolution_2019}. 

This erroneous bond-dimension growth contributes greatly to the computational complexity of contracting the tensor-network. The additional dynamics time-evolved causing more nonphysical interactions downstream, plausibly leading to further bond-dimension growth. Consider simulating with an exact matrix-product-state (MPS) tensor-network, and let $\chi(t)$ represent the maximum bond-dimension at time $t$, and let $\chi_{\max}$ be the largest value reached during the entire evolution. When a nearest-neighbour two-qubit gate is applied to an MPS in mixed-canonical form, the leading computational cost associated with the bond update and subsequent SVD re-factorization scales as $\mathcal{O}(\chi^3)$ \cite{orus_practical_2014}. Therefore, for an MPS containing $\nu$ two-site gates, the contraction cost $ C$ can be expressed as
\begin{equation}
    C\;\propto\; \sum_{m=1}^{\nu} \chi_m^{\,3} \;\le\; \nu\,\chi_{\max}^{3},
    \label{eq:Wsample}
\end{equation}
where $\chi_m$ denotes the bond dimension at the $m$-th gate application. This upper bound highlights the two principal contributors to the classical computational cost: the number of gates $\nu$ and the growth of the bond dimension $\chi(t)$ driven by entanglement.

The polynomial complexity dependence on the bond-dimension means that the erroneous growth caused by the discretization error will entail an increased time-to-solution. Thus, since the continuous TE-PAI circuits are approximations of the infinitely deep Trotter circuit they may be computationally cheaper due to requiring a smaller bond-dimension per-circuit than the Trotterization. This reduction is not guaranteed for all configurations due to the variable $\Delta$ and highly-entangling $\pi$ rotations. Instead we observe this behaviour consistently across a number of numerical experiments. 

In realistic applications of tensor-networks the exponential bond-dimension growth is the dominant factor in the complexity leading to bond-dimension truncation at a certain threshold. The truncation incurs a truncation error, recently shown to be plausibly lower for finite-$N$ TE-PAI than for Trotterization in \cite{Fredrik_2026}, and so having a lower bond-dimension growth in continuous TE-PAI leads to either a further lowered truncation error or an extended simulatable regime for the same error compared to Trotterization. As such, the lack of discretization error of continuous TE-PAI is additionally advantageous in the tensor-network context compared to Trotterization.

As a numerical demonstration \cref{fig:bond-growth} shows a system of $n=30$ spin-chain qubits with Hamiltonian $H = \sum_{k\,\in\,\text{ring}(n)} \omega_k Z_k
        + J(t)\,\Vec{\sigma}_k \cdot \Vec{\sigma}_{k+1}$.
Here, a time-dependent coupling of $J(t)=0.1\, \cos(800 \, \pi \, t)$ and $\omega_k$ are drawn uniformly from $[-1, 1]$ were simulated with $N_s=1000$ circuits of $\Delta=\pi/2^{10}$ compared to Trotterization with $N=\{10,1800,3000\}$ steps on MPS. In the upper panel continuous TE-PAI follows the canon Trotterization with high accuracy despite only using on average $\approx133,000$ gates compared to the $360,000$ Trotter-gates, whilst the shallower Trotterization simulations have a large discretization error compared to the canon irrespective of having a lower or higher gate-count than the continuous TE-PAI circuits. In the lower panel there is an exponential rise in the required bond-dimension of the shallow Trotterization simulations that is not seen in the deep Trotterization or the continuous TE-PAI circuits. This indicates that the discretization error has led to unphysical entanglement-buildup in the effective Hamiltonian of the simulation, leading to orders-of-magnitude larger bond-dimensions for both of the shallower Trotterizations than the deeper circuit or the shallow continuous TE-PAI simulation.

This high-frequency Hamiltonian was chosen to have a large discretization error for demonstration purposes. The simulations demonstrate both the gate-count savings that continuous TE-PAI can have compared to deep Trotterization, and a situation where the deep Trotterization is required to achieve a high-accuracy. Additionally the experiment demonstrates another avenue of computational advantage arising from the continuous TE-PAI circuits approximating the entanglement build-up of the system without discretization error manifested through using a low bond-dimension where comparably deep Trotter circuits experience an exponential rise. As governed by \cref{eq:Wsample} this bond-dimension reduction can exponentially reduce the computational complexity of contracting the continuous TE-PAI MPS. This advantage is nuanced by the standard practice of truncating bond-dimension in tensor-network simulations. However, when continuous TE-PAI circuits require a lower overall bond-dimension it is expected that they will experience a smaller bond-dimension truncation error than the shallow Trotterization, making the computational advantage into an accuracy one. Additionally as recently shown in \cite{Fredrik_2026} due to the aggregation of random circuit-variants in the TE-PAI protocol the truncation error is already expected to be lower than standard Trotterization.

\begin{figure}
    \centering
    \includegraphics[width=1.0\linewidth]{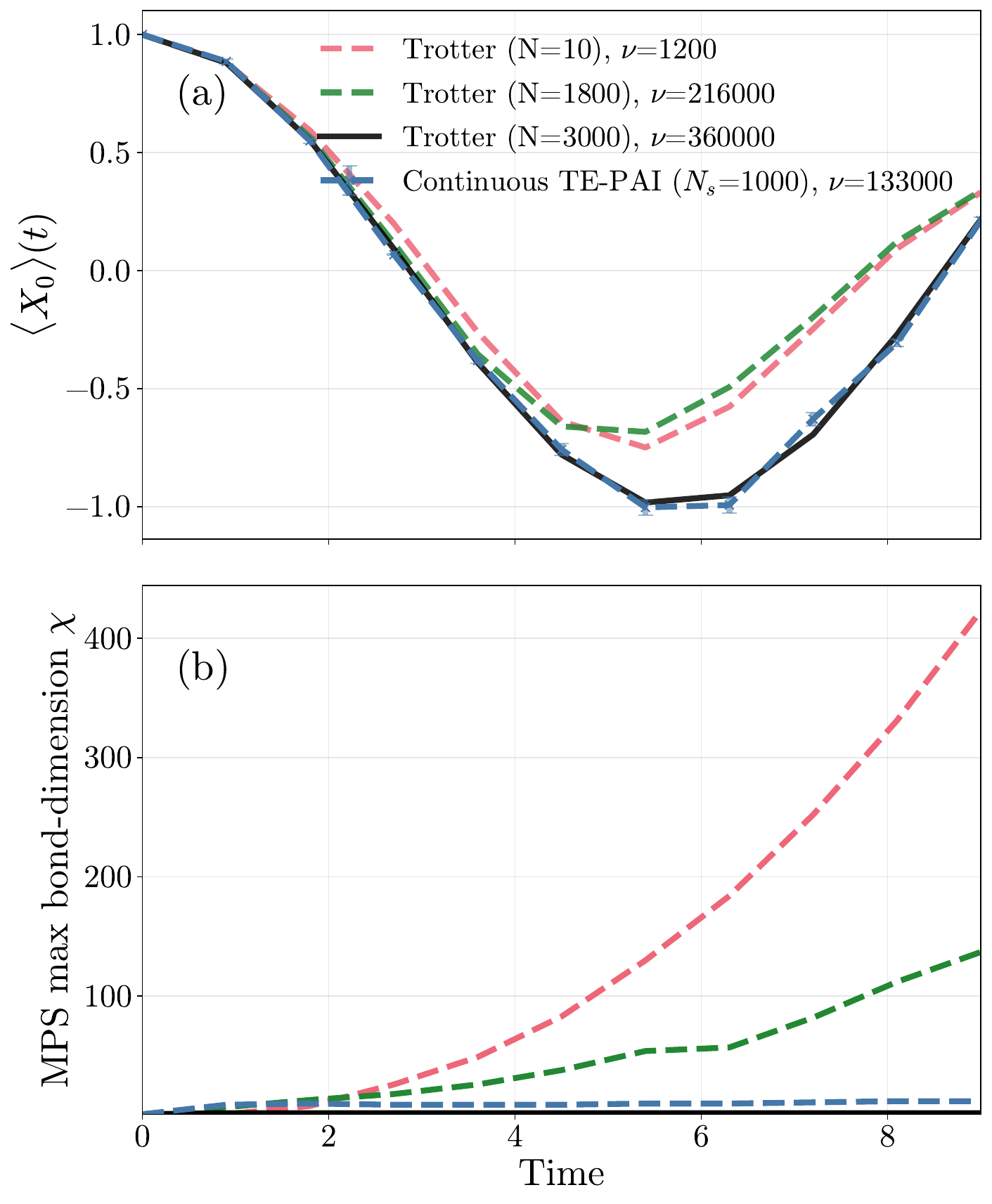}
    \caption{Continuous TE-PAI with \(N_s=1000\) circuits and \(\Delta=\pi/2^{10}\) simulates an \(n=30\) qubit spin-chain Hamiltonian with high-frequency time-dependent nearest-neighbor coupling, compared against first-order Trotterization with \(N=\{10,1800,3000\}\) on MPS. \textbf{(a): }The expectation value $\langle X_0 \rangle$ over time for the deepest canon Trotterization in black approximated by continuous TE-PAI in blue, $N=1800$ Trotter in green, and $N=10$ Trotter in red. \textbf{(b): }The maximum bond-dimension used by the simulations as a function of simulation duration. For the continuous TE-PAI circuits, this was the max between all circuits for any given time-step.}
    \label{fig:bond-growth}
\end{figure}

\section{Structure-aware variance reduction}
\label{sec:structure_aware_variance_reduction}

Continuous TE-PAI removes product-formula discretization bias, but its worst-case sampling overhead is still governed by an \(\ell_1\)-type scale of the Hamiltonian coefficients. 
Such a bound does not use commutation relations, locality, observable support, or quasiprobability sign structure. 
This section develops a structural view of the remaining variance.

The key object is the counts vector. 
It compresses a full trajectory description into the number of times each elementary letter appears. 
Algebraically, this is the canonical abelianisation of the trajectory: it forgets ordering and retains only the sampled multiset.
Conditioning on this statistic separates trajectory variance into a counting component, caused by fluctuations in which gates were sampled, and an ordering component, caused by different orderings of the same multiset. A more comprehensive formalisation is provided in \cref{app:main_theory}.

\subsection{Counts-vector decomposition}

Let \(\mathcal A=\{a_1,\ldots,a_q\}\) be a finite alphabet. 
For a word
\(\underline x=(x_1,\ldots,x_m)\), define
\[
    \mathbf N(\underline x)
    :=
    (N_1(\underline x),\ldots,N_q(\underline x)),
    \qquad
    N_i(\underline x):=\#\{r:x_r=a_i\}.
\]
The counts vector is also known as the Parikh vector
\cite{kozen2007automata}.  
It is the canonical projection from the free
noncommutative monoid generated by \(\mathcal A\) to the free commutative
monoid \cite{diekert1995book}.
Thus two words have the same counts vector when they contain the same multiset of letters.

Applying \cref{eq:conditional_risk} with \(S=\mathbf N\) gives
\begin{equation}
\label{eq:counts_variance_decomposition}
\mathcal R(\widehat\Phi)
=
\underbrace{
\mathbb E\!\left[\Var(\widehat\Phi\mid\mathbf N)\right]
}_{\mathcal R_{\mathrm{ord}}}
+
\underbrace{
\mathbb E\left\|
\mathbb E[\widehat\Phi\mid\mathbf N]
-
\mathbb E[\widehat\Phi]
\right\|_{HS}^{2}
}_{\mathcal R_{\mathrm{cnt}}}.
\end{equation}
The first term, \(\mathcal R_{\mathrm{ord}}\), is the
residual variance over different orderings of a fixed sampled multiset. The second term, \(\mathcal R_{\mathrm{cnt}}\), is the variance of the sampled abelianized dynamics across different multisets. In
other words, conditioning on \(\mathbf N\) turns a randomized trajectory
algorithm into a mixture of fixed-multiset random-order product formulas.

\subsection{Local defects: swaps and substitutions}

The two terms in \cref{eq:counts_variance_decomposition} are controlled by two different local operations. 
Assume first that the estimator is unweighted and unitary-channel-valued (as for qDRIFT),
\[
    \widehat\Phi=\Gamma(\underline X),
    \qquad
    \Gamma_a(\rho)=U_a\rho U_a^\dagger.
\]
The local defect for ordering variance is an adjacent swap.
Define
\[
    \Delta_{ab}
    :=
    \left\|
    \Gamma_a\circ\Gamma_b-\Gamma_b\circ\Gamma_a
    \right\|_{HS}.
\]
For unitary channels,
\[
    \Delta_{ab}^{2}
    =
    2\left(
    d^2
    -
    \left|
    \Tr(U_a^\dagger U_b^\dagger U_aU_b)
    \right|^2
    \right),
\]
and \(\Delta_{ab}=0\) precisely when the two single-letter channels commute.
Indeed, if two words differ only by an adjacent transposition,
\[
    \underline x=uabv,
    \qquad
    \underline y=ubav,
\]
then
\[
    \Gamma(\underline y)-\Gamma(\underline x)
    =
    \Gamma_v\circ
    (\Gamma_b\circ\Gamma_a-\Gamma_a\circ\Gamma_b)
    \circ\Gamma_u .
\]
Left and right composition by unitary channels are Hilbert--Schmidt isometries, so this edge increment has norm \(\Delta_{ab}\).

The corresponding defect for for counting variance is due to single-letter substitutions. 
Define
\[
    \Sigma_{ab}
    :=
    \|\Gamma_a-\Gamma_b\|_{HS}.
\]
For unitary channels,
\[
    \Sigma_{ab}^{2}
    =
    2\left(
    d^2
    -
    \left|
    \Tr(U_a^\dagger U_b)
    \right|^2
    \right),
\]
and \(\Sigma_{ab}=0\) precisely when the two channels are identical.  
If \(\underline x=uav\) and \(\underline y=ubv\), then
\[
    \Gamma(\underline y)-\Gamma(\underline x)
    =
    \Gamma_v\circ(\Gamma_b-\Gamma_a)\circ\Gamma_u ,
\]
whose norm is \(\Sigma_{ab}\).

Thus the counts-vector decomposition separates two local mechanisms:
\[
\begin{array}{c|c}
\text{ordering variance} & \text{counting variance} \\
\hline
\text{swap }uabv\leftrightarrow ubav
&
\text{substitution }uav\leftrightarrow ubv
\\[1mm]
\Gamma_a\Gamma_b-\Gamma_b\Gamma_a
&
\Gamma_a-\Gamma_b
\\[1mm]
\Delta_{ab}
&
\Sigma_{ab}.
\end{array}
\]
The sample-space geometry behind this decomposition is illustrated in
\cref{fig:variance graph}. 
The counts vector first partitions the full word space into integer-simplex points representing multisets. 
Each multiset then carries an internal permutation graph.
Stratifying by counts removes variance from choosing the multiset randomly (deterministically allocating samples across the outer graph); the residual is variance over orderings inside each multiset (random sampling on the internal graph).

\begin{figure*}
    \centering
    \includegraphics[width=0.9\linewidth]{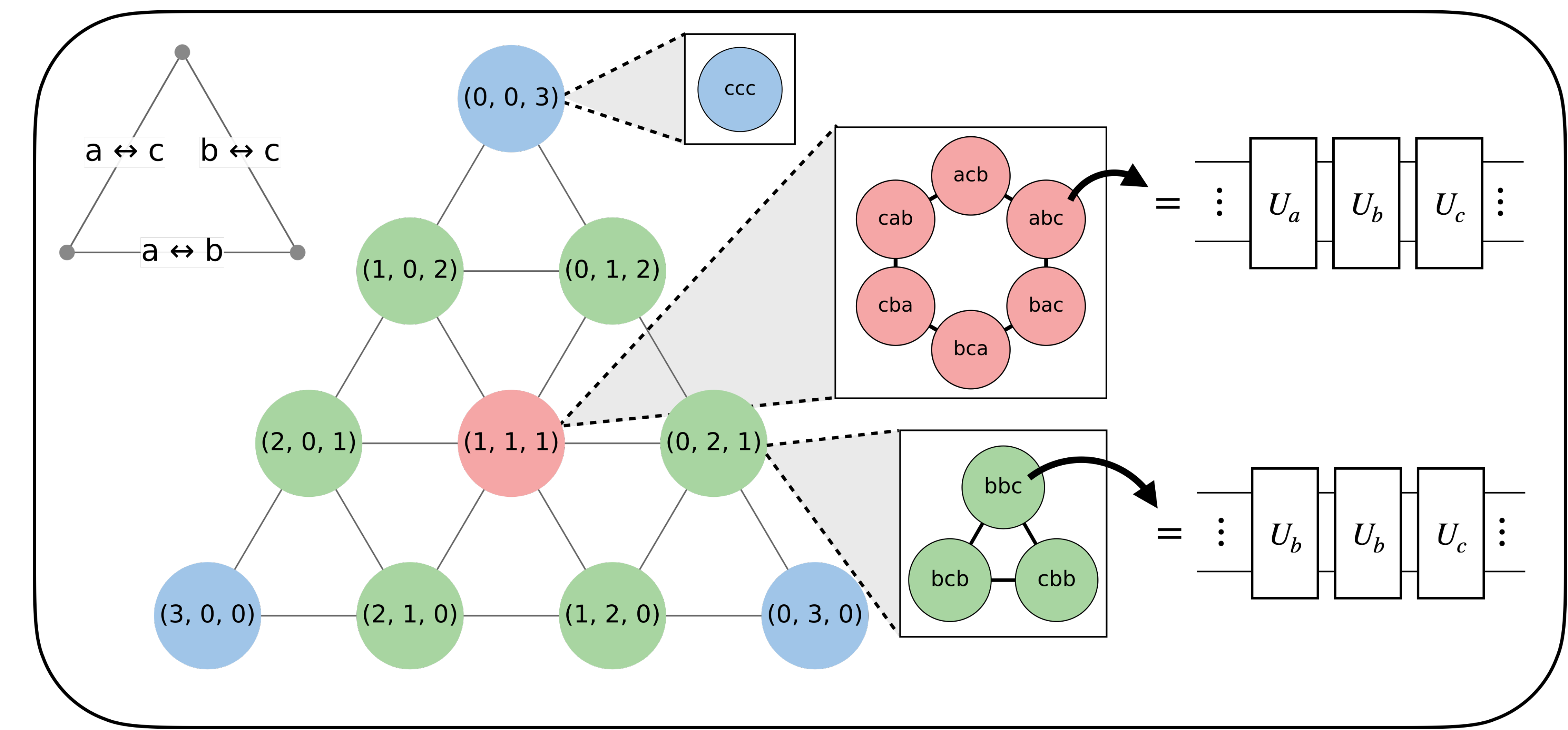}
    \caption{Partitioning of the sample space for words of length 3 made from an alphabet \(\mathcal A=\{a,b,c\}\). 
    The \(3^3=27\) words are first
    grouped by the counts vector into \(\binom{5}{2}=10\) multisets, represented as points on an integer simplex.
    Edges between simplex points correspond to single-letter substitutions. 
    Each multiset has an internal graph whose
    nodes are the distinct permutations of that multiset and whose edges are adjacent transpositions.
    Counts-vector stratification removes variance associated with the outer substitution graph; the residual variance is due to random sampling of orderings on the internal permutation graphs.}
    \label{fig:variance graph}
\end{figure*}

\subsection{Structural bounds}
The local swap and substitution identities can be lifted to global variance bounds using weighted graph Poincar\'e inequalities. 
The full graph construction and proofs are given in \cref{app:main_theory}; here we state the resulting schematic bounds.  
The argument uses standard tools from random-walk and random-transposition theory \cite[Sec.~16.1.4]{levin_markov_2017,caputo_proof_2010}.

For each count sector, equip the set of words with that count vector with an adjacent-swap graph. 
Then there are nonnegative law-dependent weights \(W^{\mathrm{ord}}_{ab}\) such that 
\begin{equation}
    \mathcal R_{\mathrm{ord}}
    =
    \mathbb E\!\left[\Var(\widehat\Phi\mid\mathbf N)\right]
    \le
    \sum_{a<b}W^{\mathrm{ord}}_{ab}\Delta_{ab}^{2}.
    \label{eq:schematic_order_bound}
\end{equation}
Full details can be found in the Appendix, \cref{prop:poincare_counts_residual_law_agnostic}. 
Similarly, on the graph of count vectors connected by elementary substitutions, there are nonnegative weights \(W^{\mathrm{sub}}_{ab}\) such that
\begin{equation}
    \mathcal R_{\mathrm{cnt}}
    \le
    \sum_{a<b}W^{\mathrm{sub}}_{ab}\Sigma_{ab}^{2}.
    \label{eq:schematic_count_bound}
\end{equation}
The exact weights depend only on the trajectory law and on the chosen
comparison graphs.

Equations \cref{eq:schematic_order_bound,eq:schematic_count_bound} show that the ordering and counting components are controlled by different physical properties. 
The ordering term is governed by noncommutativity of sampled channels; the counting term is governed by distinguishability of the possible letters. In particular, if all elementary channels commute, then \(\Delta_{ab}=0\) for all \(a,b\), the realized channel factors through the counts vector,
$
    \Gamma(\underline x)=\Gamma_{\mathrm{ab}}(\mathbf N(\underline x)),
$
and the within-count residual vanishes as
\[
    \mathcal R(\widehat\Phi,\mathbf N)=0.
\]
In a commuting problem, fixing the sampled multiset fixes the trajectory channel completely.

\subsection{A qDRIFT example}
The decomposition is especially transparent for qDRIFT. 
Consider a time-independent Hamiltonian
\[
    H=\sum_{a\in\mathcal A}h_aP_a,
    \ 
    \lambda:=\sum_{a}|h_a|,
    \ 
    p_a:=\frac{|h_a|}{\lambda},
    \ 
    s_a:=\operatorname{sgn}(h_a).
\]
A length-\(m\) qDRIFT trajectory samples letters
\(X_1,\ldots,X_m\) independently from \(p_a\), and realizes letter \(a\) as
\[
    U_a
    =
    \exp\left(-i\,s_a\frac{\lambda T}{m}P_a\right),
    \qquad
    \Gamma_a(\rho)=U_a\rho U_a^\dagger .
\]
The mean channel is
\[
    \overline\Phi_m
    =
    \mathbb E[\Gamma(\underline X)]
    =
    \left(\sum_{a\in\mathcal A}p_a\Gamma_a\right)^m,
\]
which generally differs from \(\mathcal U_T\). 
Thus qDRIFT has both mean-channel bias and sampling variance:
\[
\mathbb E\left|
\frac1{N_s}\sum_{r=1}^{N_s}L(\widehat\Phi_r)
-
L(\mathcal U_T)
\right|^2
=
\left|L(\overline\Phi_m-\mathcal U_T)\right|^2
+
\frac1{N_s}\Var(L(\widehat\Phi)).
\]
Since the letters are sampled independently, the counts vector has multinomial law,
\[
    \mathbf N\sim\operatorname{Multinomial}(m,p),
\]
and, conditional on \(\mathbf N=\mathbf n\), the trajectory is uniformly distributed over all permutations of the corresponding multiset.
The counts fix the empirical Hamiltonian
\[
    H_{\mathrm{emp}}(\mathbf n)
    :=
    \frac{\lambda}{m}
    \sum_{a\in\mathcal A}n_a s_aP_a,
    \qquad
    \mathbb E_{\mathbf N}[H_{\mathrm{emp}}(\mathbf N)]=H.
\]
Hence the counting component is empirical-Hamiltonian sampling noise, while the ordering component is random-permutation product-formula noise at fixed empirical Hamiltonian.

This perspective is consistent with the concentration results of Ref.~\cite{chen_concentration_2021}. 
The mean qDRIFT channel has weak-error scaling controlled by the depth \(m\), while a single trajectory fluctuates around that mean.
In commuting examples the ordering variance vanishes, so the single-trajectory fluctuations are purely count fluctuations.

\begin{figure*}
    \centering
    \includegraphics[width=0.9\linewidth]{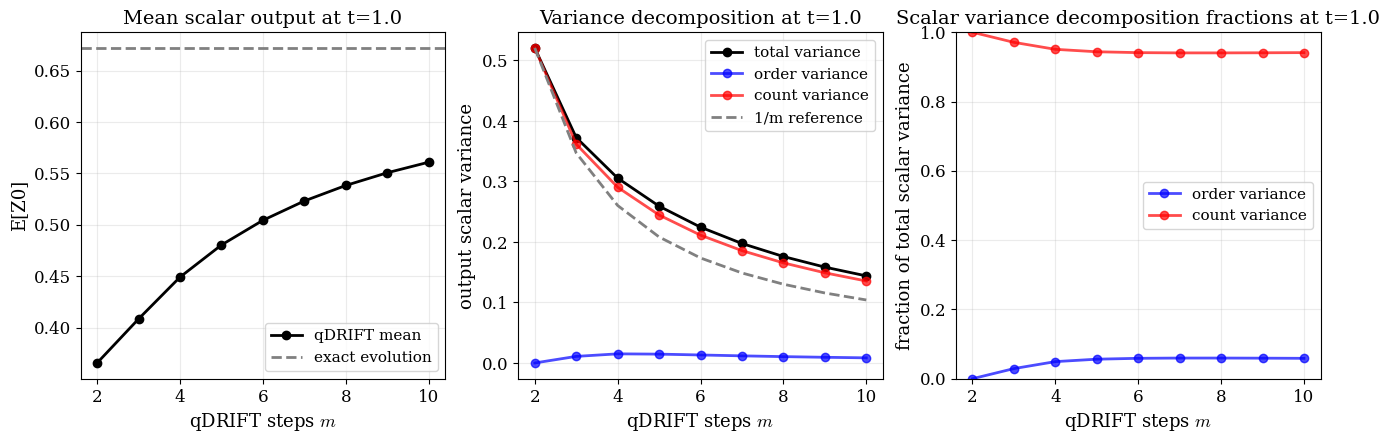}
    \caption{A small qDRIFT example on the 2-qubit Hamiltonian
    \(H=\frac12(X_0+X_1)+Z_0Z_1\) at \(t=1\), varying the qDRIFT step number \(m\), and estimating \(\langle Z_0\rangle\). 
    The left panel shows convergence of the qDRIFT mean to the exact value. 
    The middle panel shows the scalar trajectory-variance decomposition into the residual ordering variance after counts-vector stratification and the counting variance removed by that stratification.
    The right panel shows the same decomposition as a fraction of the unstratified variance. 
    The plotted variance excludes measurement shot noise, which would add a further single-shot contribution and dilute the realized reduction
    \cite{Dai_2025}.
    Variances were obtained by exhaustive enumeration of all \(3^m\) qDRIFT configurations. We note that the specific values of the decomposition are strongly problem and observable dependent.}
    \label{fig:qdrift_example}
\end{figure*}

\subsection{Connection to operator spreading}

The ordering component also has a useful connection to operator spreading. This connection is operational, not thermodynamic: the ordering variance is not a universal chaos or thermalization diagnostic. 
Rather, it measures how sensitive a particular randomized estimator and observable are to noncommuting swaps in the sampled trajectory. For a scalar output 
\[
    Y(\underline x)
    =
    \Tr[O\,\Gamma(\underline x)(\rho_0)],
\]
consider two words differing by an adjacent swap, \(\underline x=uabv\) and \(\underline y=ubav\).
In the Heisenberg picture,
\[
    Y(uabv)-Y(ubav)
    =
    \Tr\!\left[
        \rho_u
        (\Gamma_a^\ast\Gamma_b^\ast-\Gamma_b^\ast\Gamma_a^\ast)(O_v)
    \right],
\]
where
\[
    \rho_u=\Gamma_u(\rho_0),
    \qquad
    O_v=\Gamma_v^\ast(O).
\]
For unitary elementary channels this difference is controlled by the OTOC-like squared commutator
\[
    \|[O_v,K_{ab}]\|_2^2,
    \qquad
    K_{ab}=U_bU_aU_b^\dagger U_a^\dagger .
\]
Thus the local swap terms appearing in the ordering variance are algorithmically sampled commutators between a Heisenberg-dressed observable and a local swap probe.
This is closely related in form to squared-commutator and OTOC diagnostics of operator growth \cite{chen2018operatorscrambling,chen2023speed}, but the ensemble is different: standard random-circuit OTOC results study operator spreading under random physical dynamics \cite{nahum2018operator,von2018operator}, whereas here the target dynamics is fixed and the randomness comes from the qDRIFT sampling procedure.

The interpretation is qualitatively useful. 
If \(O_v\) remains local or effectively abelian, most swap probes commute with it and counts-vector stratification removes most trajectory variance. 
If \(O_v\) spreads into many noncommuting operator sectors, more swaps become visible and the ordering component can become significant.

\subsection{Scalability and signed estimators}

The full counts vector is a canonical diagnostic but usually not a scalable stratification statistic.
For fixed word length \(m\) and alphabet size \(q=|\mathcal A|\), its support has size
\(
    \binom{m+q-1}{q-1}.
\)
When \(q\) scales with the number of Hamiltonian terms, exact counts-vector stratification is impractical.
In numerical experiments we therefore use statistics motivated by the same decomposition, but necessarily coarsened so as to be practical.

For signed estimators such as TE-PAI, the same variance decomposition applies to
\[
    \widehat\Phi(\underline X)=G(\underline X)\Gamma(\underline X).
\]
However, to interpret the within-stratum residual as physical ordering variance, the statistic should include the variables that determine the weight.
If \(G\) is constant on each stratum \(S=s\), then
\[
    \Var(G\Gamma\mid S=s)
    =
    G_s^2\Var(\Gamma\mid S=s).
\]
If not, the residual mixes ordering fluctuations with quasiprobability-weight fluctuations.

For fixed-angle continuous TE-PAI,
\[
    g_\omega
    =
    G_0(-1)^{N_\pi(\omega)},
    \qquad
    G_0
    =
    \exp\!\left[
        2\overline{\|c\|_1}T\tan(\Delta/2)
    \right],
\]
so the sign is determined by the parity of the number of \(\pi\)-rotation events. 
Practical TE-PAI statistics should therefore fix the quasi-probability sign, since then the remaining residual variance is governed by the noncommutativity of the sampled physical rotations, exactly as in the unweighted setting up to the factor \(G_0^2\).

\subsection{Application to continuous TE-PAI}
\label{sec:tepai_variance_reduction_application}

Continuous TE-PAI is the main setting in which the preceding variance reduction is directly algorithmically useful. 
Its weighted estimator satisfies
\[
    \mathbb E_\omega[g_\omega\mathcal U_\omega]=\mathcal U_T,
\]
so for every observable functional \(L\),
\[
\mathbb E\left|
\frac1{N_s}\sum_{r=1}^{N_s}
L(g_{\omega_r}\mathcal U_{\omega_r})
-
L(\mathcal U_T)
\right|^2
=
\frac1{N_s}
\Var\!\left(L(g_\omega\mathcal U_\omega)\right).
\]
Any variance reduction therefore translates directly into reduced sample complexity. A continuous TE-PAI trajectory is a marked point process
$
    \omega=(t_m,k_m,\ell_m)_{m=1}^{M},
$
with event time \(t_m\), Pauli label \(k_m\), and type label \(\ell_m\in\{0,1\}\) distinguishing \(\Delta\)-rotations from \(\pi\)-rotations.

The full count statistic is too large, so we instead use finite,
observable-adapted projections of the counts vector. The statistics below are
chosen to be:
\begin{enumerate}[leftmargin=*,itemsep=2pt]
    \item sign-aware, so that quasiprobability signs are fixed within strata;
    \item local or observable-adapted, so that they focus on gates visible to the measured observable;
    \item small enough that conditional sampling and finite-budget allocation remain practical.
\end{enumerate}
Thus the counts-vector theory supplies the ideal abelian reference point, while the numerical TE-PAI procedures use scalable approximations to that reference.

\section{Numerical simulation}
\label{sec:numerics}

In this work we do not aim to determine the optimal variance-reduction protocol for unbiased randomized Hamiltonian simulation. This will in general depend on the details of the system being simulated, the resources available to do so, and the purpose and required accuracy of the simulation. Instead, this section demonstrates the efficacy of structure-aware stratified-sampling for two distinct statistics applied to two different Hamiltonians showcasing different simulation methods and parameters.

In \cref{sec:pi-stratification} we outline a statistic centred on the high impact of $\pi$-rotations in continuous TE-PAI as a simple, scalable, and modular statistic capable of large variance reductions. Subsequently, in \cref{sec:obs-stratification} we present a more structure-aware statistic based on the counts-vector accounting for the locality of the observable.

% TODO: add asymptotic 
\subsection{$\pi$-stratification}
\label{sec:pi-stratification}

From the Dyson-series representation of continuous TE-PAI outlined in \cref{sec:dyson} one can construct a stratification statistic based on $\pi$-rotations. These rotation choices are highly unlikely for any given gate, and have an outsized effect on the time-evolution channel when applied. This leads to $K+1$ well-defined strata when considering up to $K$ $\pi$-flips. This statistic is inherently tractable as for a given simulation the probability of encountering an increasing number of $\pi$-rotations decays exponentially. Additionally, it is fully compatible and often synergizes with a number of other statistics that are based on examining high-impact events on an observable's trajectory.

This statistic generates strata that can be subdivided by considering which factors modulate the effect of the $\pi$-rotations. An example is considering the light-cone of the $\pi$-rotation and its intersection with the qubit we are targeting with the simulated observable. This will clearly depend on the locality of the Hamiltonian and observable, and generally back-propagation through the circuit becomes exponentially costly due to the growth of operator support and entanglement. For practicality we fix a generational depth $d$ as the number of generations in which we check the light-cone of the $\pi$-rotation for our target qubit. For simplicity we only track the total number of $\pi$-rotations that occur near ($N$) or far ($F$) from our target qubit(s), which makes the triangle number of $K+1$, $T(K+1)=\sum_{k=1}^{K+1}k$, strata. This distinction will likely have an especially large impact for large system sizes, scaled down factors that spread the impact of the $\pi$-rotation like Hamiltonian and observable locality and simulation time.

Furthermore, the time at which a $\pi$-flip is encountered is likely to have an effect on our observable, as early encounters may manifest differently in the final expectation value than later ones. As such we further subdivide our simulation into $S$ time-steps, assigning one to each $\pi$-rotation. The separate possibility of encountering multiple $\pi$-rotations within the same time-step can be omitted for simplicity as the effect compared to a single encounter may be small and the likelihood can be tuned down by increasing the number of time-steps. Counting the number of possible $\pi$-rotation combinations and their possibility for being near and far from the target qubit gives the following stratum count:

\begin{align}
N(K,S)&=\sum_{\pi\text{-rotations}}(\text{\# combinations})\times (\text{\# localities}) \nonumber \\
N(K,S)&=\sum_{k=0}^{\min(K,S)} \binom{S}{k} \ 2^k
\end{align}
In the special case of $K\geq S$ this reduces nicely to $3^S$ stratum.

The result of these three levels of $\pi$-rotation variance reduction is shown in \cref{fig:stratification-histogram} compared to naive sampling. For this numerical demonstration a $n=12$ qubit spin-chain Hamiltonian was simulated for a duration of $T=2$ using continuous TE-PAI circuits with $\Delta=\pi/2^7$. The experiment considered up to $K=2$ $\pi$-rotations, with near ($N$) or far ($F$) locality and subdivided the stratum into $S=2$ time-steps. In subplot \textbf{(a)} the naive histogram is shown which has a per-circuit estimator standard deviation of $\sigma_{\text{tot}}=1.5$. In \textbf{(b)} differentiating the 0-, 1-, and 2-$\pi$ cases gives a markedly lower variance of $\sigma_{\text{tot}}=0.86$, but apart from $0$-$\pi$ the strata are better approximated by bimodal Gaussian probability functions indicating a substantial amount of residual variance. For \textbf{(c)} subdividing based on the light-cone yields a lower $\sigma_{\text{tot}}=0.80$ where now the stratum are reasonably well approximated by a Gaussian fit. Finally for \textbf{(d)} separating the simulation into two time-steps yields a slight improvement,  $\sigma_{\text{tot}}=0.78$, but with little difference between strata that differ only by time step. This weak effect is likely explained by the short simulation duration of this experiment and the crude number of time-steps considered. The parameters for this experiment were chosen to aid in the visualisation of the different stratum through their histograms and we expect more sophisticated configurations to yield more dramatic results in larger-scale simulations.

\begin{figure*}[t]
    \centering
    \includegraphics[width=1.0\linewidth]{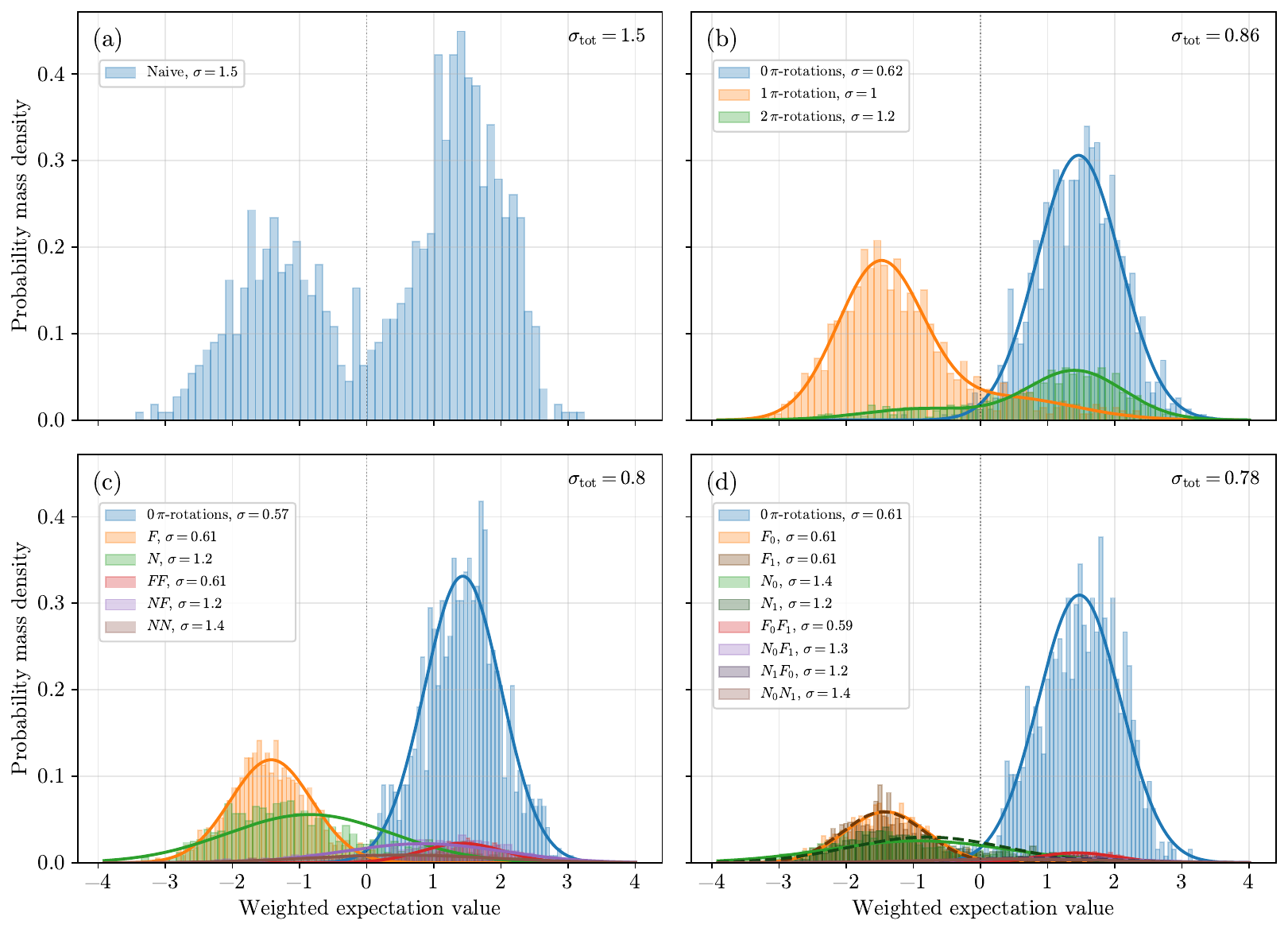}
    \caption{%
    Per-circuit distributions of the (weighted) observable estimator
    $w(\mathcal{C})\,\langle O\rangle_{\mathcal{C}}$ for $\langle X_{0}\rangle$ on a
    $12$-qubit time-dependent spin chain ($J=1$, $T=2$, $\Delta=\pi/128$,
    $N_{\text{s}}=10^{3}$ samples per stratum), compared across four sampling
    strategies. Bars show the empirical histogram of each (sub)stratum, scaled by
    its analytic stratum probability $p_{s}$, so the panel-wide pointwise sum
    approximates the mixture density targeted by the corresponding estimator.
    Solid (and dashed, panel~(d)) curves are parametric fits described below.
    The bold quantity in the top-right of each panel is the resulting per-circuit
    estimator standard deviation under proportional allocation,
    $\sigma_{\mathrm{tot}}=\sqrt{\sum_{s} p_{s}\sigma_{s}^{2}}$, which collapses
    to the raw sample standard deviation in the unstratified case~(a).
    \textbf{(a)}~Naive (unstratified) sampling, raw histogram only.
    \textbf{(b)}~Stratification by total $\pi$-flip count $K$ ($K_{\max}=2$),
    color-coded by $K$, with a bimodal Gaussian fit on each $K\!\geq\!1$ stratum
    and a single Gaussian on $K=0$. \textbf{(c)}~$(K,i)$~observable-locality
    substrata at depth $d=2$, where $i$ counts $\pi$-flips on Hamiltonian terms
    within $d$~hops of the observable; legend symbols denote near (\textsf{N}) and
    far (\textsf{F}) flips, so e.g.\ \textsf{NF} is the $K=2,\,i=1$ substratum.
    \textbf{(d)}~As~(c) but with the additional refinement of $S=2$ time buckets;
    substrata that share the same locality signature $(K,K_A,K_B)$ are drawn in
    the same hue, and substrata that differ only by which time bucket each flip
    occupies are distinguished by a darker shade and dashed fit line, with a
    subscript on each \textsf{N}/\textsf{F} indicating its bucket index.
    A single Gaussian is fitted to every substratum in panels~(c)--(d).
    The successive reduction in $\sigma_{\mathrm{tot}}$ from~(a) to~(d) quantifies
    the variance reduction obtained by progressively finer stratification.}
    \label{fig:stratification-histogram}
\end{figure*}

In \cref{fig:stratification-results} \textbf{(b,d,e)} we demonstrate the results of $\pi$-stratification applied to a $n=30$ spin-chain Hamiltonian with $J(t)=\cos(99\pi t)$ and on-site couplings chosen uniform randomly from $\omega\in[-1,1]$ simulating $\langle X_0 \rangle$. We used $N_s=1000$ continuous TE-PAI circuits with $\Delta=\pi/2^{10}$ for a total simulation-duration of $T=1$, stratified based on $K=2$ accounting for the light-cone of the $\pi$-rotation up to $d=2$ levels and subdividing into $S=2$ time-segments, yielding $3^2=9$ strata. Tracking the strata allocation over time in \textbf{(b)}, we can see that at the start the $K=0$ stratum is dominant but stratum involving more $\pi$-rotations have a probabilistic weight increasing with time. Furthermore for a given $K$ the randomly inserted $\pi$-rotations are more likely to be far from the single target qubit than close to it as judged by only 2 generations of causal influence.

The effect of the stratification is shown in \cref{fig:stratification-results} \textbf{(d)} where we track the ratio of naive versus stratified sampling for two error metrics, the root-mean square error (RMSE) and the standard error (SE), for an equal number of continuous TE-PAI circuits. The RMSE is more susceptible to uninformative noise but reflects the true accuracy of the simulation compared to a deep Trotter simulation whereas the SE more consistently shows the variance of the two sampling protocols. In either case the result is a ratio of $\approx0.2$ at the final time, implying a reduction in the variance via stratification of around $80\%$. This reduced variance is visualised in \textbf{(e)} where we plot the naive probability density distribution alongside its stratified counterpart such that their relative width and shared mean indicates the increased precision achieved when using the stratified sampling.

\begin{figure*}
    \centering
    \includegraphics[width=1.0\linewidth]{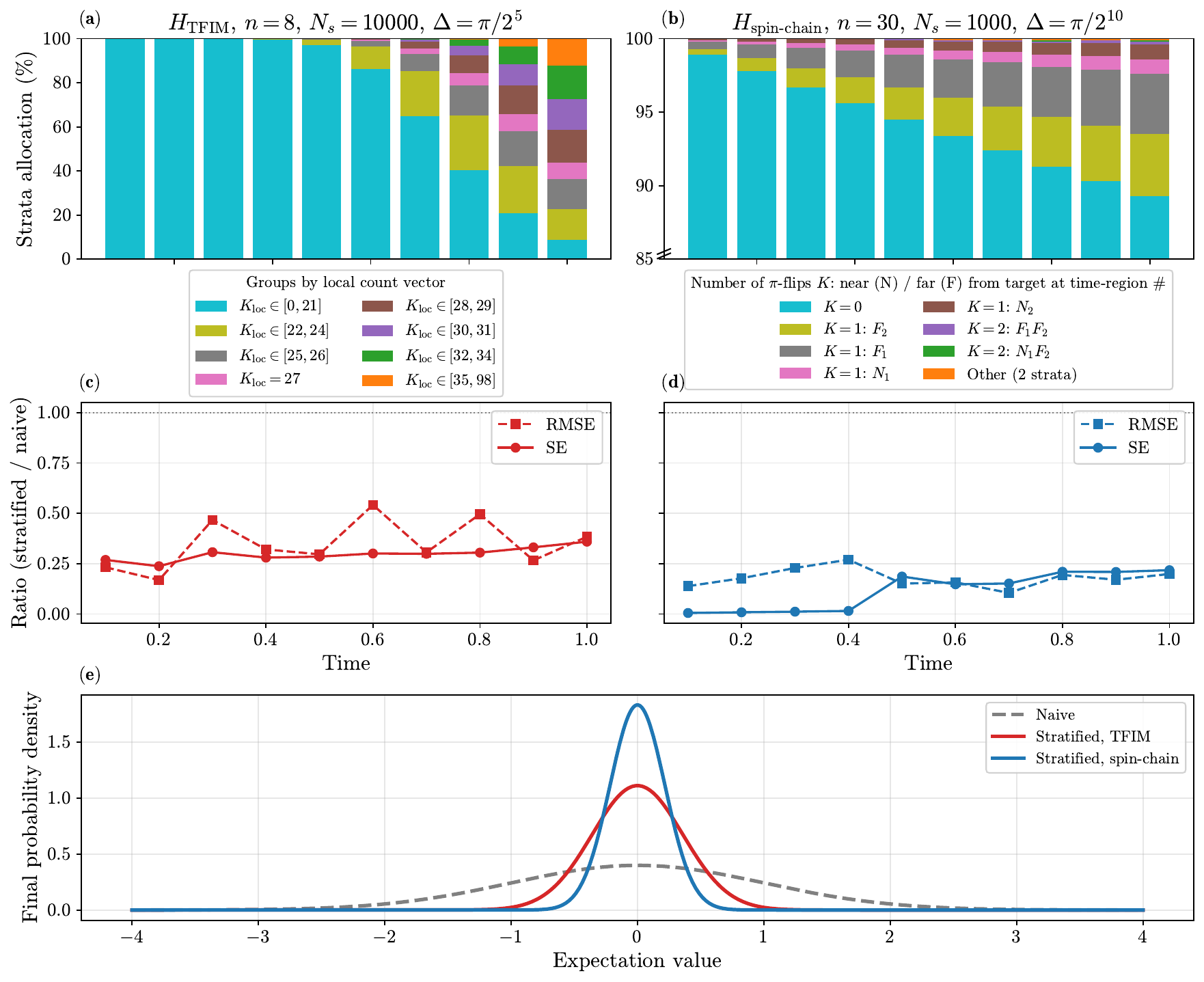}
    \caption{ Stratified versus naive continuous TE-PAI for an 8-qubit transverse-field Ising chain (periodic, $J=0.5$, $h=0.4$, $\Delta=\pi/2^5$, $N_s=10\,000$ circuits) and a 30-qubit time-dependent Heisenberg-like chain ($J(t)=\cos(99\pi t)$, $\Delta=\pi/2^{10}$, $N_s=1000$) run with an MPS tensor-network simulation. The TFIM was stratified based on the local-count-parity of the $\Delta$-rotations involving the target qubit plus the parity of $\pi$-flips on the remaining qubits. $\sim$76 000 strata were binned into eight $K_\text{loc}=\sum_i k_i$ probability octiles for visualisation. The spin-chain was stratified with $\pi$-flip number (up to $K=2$) subdivided based on observable-locality ($N$ / $F$ for near / far from the target qubit as defined by $d=2$ generations backpropagation) and time-segmentation into $s=2$ segments (subscript number).  \textbf{(a,b)}: Stratum allocation percentage as a function of time. \textbf{(c,d)}: Stratified-to-naive ratio of root-mean-square error (RMSE, dashed) and standard error (SE, solid) over time; values below the dotted unit line quantify the precision gain. \textbf{(e)}: Final-time normalized Gaussian probability densities for the naive estimator (gray dashed) and the two stratified estimators (TFIM red, spin-chain blue) plotted against the centered expectation value; the height ratio equals the precision gain $\sigma_\text{naive}/\sigma_\text{strat}$.}
    \label{fig:stratification-results}
\end{figure*}

\subsection{Observable-aware stratification}
\label{sec:obs-stratification}

The framework of \cref{sec:observable_adapted_stratification} suggests that for a fixed-locality observable, much of the variance in the per-circuit estimator can be removed by stratifying only over count coordinates inside a small neighborhood of the observable support. We test this picture on a non-commuting Hamiltonian where the locality bound of \cref{prop:observable_residual_local_counts} is non-trivial: a periodic transverse-field Ising chain
\[
H
=
-J\sum_{j=1}^{n}Z_jZ_{j+1}
-h\sum_{j=1}^{n}X_j,
\]
with \(n=8\), \(J=0.5\), \(h=0.4\), simulated with continuous TE-PAI at \(\Delta=\pi/2^5\) for total time \(T=1\) using \(N_s=10^4\) circuits. The target observable is the single-site Pauli \(X_3\) on the \(\ket{+}^{\otimes n}\) initial state, and a deep Trotter circuit is used to track accuracy.

For this Hamiltonian the visible alphabet of \(X_3\) consists exactly of the three terms whose support touches site~\(3\): the transverse field \(X_3\) and the two adjacent bonds \(Z_2Z_3\), \(Z_3Z_{4}\). We therefore choose the alphabet
\[
\mathcal A_{R,\ell}
=
\{X_3,\,Z_2Z_3,\,Z_3Z_{4}\},
\]
and use the type-2 (\(\Delta\)-rotation) count vector \(\mathbf N_{R,\ell}\) over this retained alphabet as the locality-truncated count statistic. Unlike the commuting Ising example of \cref{sec:observable_adapted_stratification}, the field and bond letters in \(\mathcal A_{R,\ell}\) do not commute, so \(\mathbf N_{R,\ell}\) carries a non-zero local order residual; the field term is included nonetheless because it remains within the operator-spreading light cone of \(X_3\) and contributes to the locality error \(\epsilon_{R,\ell}\) when omitted.

Naively conditioning only on the local count vector leaves a substantial outside-letter residual: type-3 (\(\pi\)-rotation) events on terms outside \(\mathcal A_{R,\ell}\) are rare per term but, summed over the chain, occur frequently and carry the dominant signed weight \(g(W)=(-1)^{N_3(W)}\). To absorb the bulk of this signed contribution at negligible support cost, we augment the local count vector with a single bit, the global parity
\[
\pi_{\text{out}}(W)
:=
\sum_{a\notin\mathcal A_{R,\ell}}N_{a,3}(W)\bmod 2,
\]
and stratify on the joint statistic \((\mathbf N_{R,\ell},\pi_{\text{out}})\). Because the local count distribution factorizes coordinate-wise (Poisson with rate \(2|h_a|T/\sin\Delta\)) and the outside parity has a closed-form Bernoulli law, the joint stratum probabilities are analytically available. Truncating each Poisson tail at total mass \(\epsilon_{\rm trunc}=10^{-8}\) yields a retained support of \(\sim 7.6\times 10^{4}\) strata, controlled entirely by the size of the local neighborhood and the truncation tolerance, rather than by the number of Hamiltonian terms. 

Sampling within each stratum is performed by drawing the local type-2 counts conditionally fixed, the outside type-2 counts from their unconditional Poissons, and the outside type-3 counts from a Poisson conditioned on parity \(\pi_{\rm out}\); times and orderings are then placed uniformly. The combined estimator is unbiased on the retained support, with a residual bias bounded by the omitted-mass tail \(\|g\|_\infty\,(1-\sum_{r}p_r)\), which remains below \(10^{-6}\) at all snapshots considered.

Stratification reduces the per-circuit standard error monotonically over the simulation interval: at \(T=1\) the local count/parity estimator achieves a standard error of \(7.6\times 10^{-2}\), compared to \(1.3\times 10^{-1}\) for naive sampling of the same TE-PAI distribution—a precision gain of \(\sigma_{\rm naive}/\sigma_{\rm strat}\approx 1.7\), corresponding to a roughly threefold reduction in the number of sampled circuits required to reach a fixed accuracy. Earlier snapshots show even larger ratios, with \(\sigma_{\rm strat}/\sigma_{\rm naive}\approx 0.3\) at \(T=0.1\), reflecting the fact that at short times the locality error \(\epsilon_{R,\ell}\) is small and the joint statistic captures essentially all of the per-circuit variability. As \(T\) grows, operator spreading injects variance from terms outside \(\mathcal A_{R,\ell}\) which the locality-truncated statistic cannot resolve, and the gain shrinks; this is exactly the locality-error/order-residual tradeoff predicted by \cref{prop:observable_residual_local_counts}.

These two design choices, the geometric locality criterion for selecting \(\mathcal A_{R,\ell}\) and the augmentation by a single global parity bit, together demonstrate that observable-aware stratification can be implemented with a stratum count that is independent of system size, while retaining a substantial fraction of the variance reduction that a full count-stratification would provide. The strategy is modular: enlarging \(\mathcal A_{R,\ell}\) to a wider light-cone neighborhood reduces \(\epsilon_{R,\ell}\) at polynomial cost in stratum support, and additional global statistics (such as further coarse parity or count moments on outside terms) can be appended to absorb residual outside-letter variance without affecting the local stratification. The result is a stratification design whose cost is set by observable locality and desired precision rather than by the global complexity of the simulated Hamiltonian.

The effect of the locality stratification is shown in \cref{fig:stratification-results} \textbf{(c)} where we track the same RMSE and SE ratios for an equal number of continuous TE-PAI circuits. Panel \textbf{(a)} shows the corresponding stratum allocation over time: at early times nearly all sampling effort concentrates within the lowest local-count octile $K_\text{loc}\in[0,21]$, where the local letters are sparse and the typical trajectory carries few $\Delta$-rotations on the retained alphabet; as the simulation progresses the allocation distributes across progressively higher octiles, reflecting the growing local activity captured by the locality-truncated count vector. In either error metric the ratio is observed to stay below $\approx 0.4$ at the final time, implying a reduction in the variance via stratification of around $60\%$. This reduced variance is visualised in \textbf{(e)} where the stratified TFIM probability density (red) is plotted alongside its naive counterpart (gray dashed), the narrower red curve at shared mean indicating the increased precision achieved through observable-aware stratification.

\label{sec:obs-stratification}

% \begin{figure}
%     \centering
%     \includegraphics[width=1.0\linewidth]{paper_figures/strata_distribution.png}
%     \caption{Histogram of weighted expected values at simulation time $T=1$ for a TFIM dynamics simulation on $n=8$ qubits with $J=1$, $h=0.8$, PAI parameter $\Delta = \pi/32$, and taking $N = 2048$ circuit samples overall. The stratification used was a count of total $\pi$ flips in the full time span.}
%     \label{fig:placeholder}
% \end{figure}

\section{Discussion}
There remains substantial freedom in the choice of stratification statistic. The optimal statistic is generally problem dependent, and can be tailored to the Hamiltonian structure, the target observable, the initial state, or the dominant source of sampling fluctuation. Optimizing this choice is an important direction for further variance reduction. This perspective is consistent with a broader trend in Trotter-based Hamiltonian simulation, where problem-specific structure is exploited to reduce simulation cost. Examples include low-rank factorizations of electronic-structure Hamiltonians~\cite{motta2021lowrank}, time-dependent product formulas that exploit separated energy scales~\cite{bosse2025efficient}, Trotterization for the Heisenberg model~\cite{yang2025symmetryaware}, local-symmetry-based Hamiltonian partitioning for Trotter decompositions~\cite{negishi2026beyond}, and improved Trotter efficiency for low-energy initial states~\cite{lowEnergyTrotter2025}. These results suggest that both Hamiltonian structure and the properties of the initial state can provide useful information not only for constructing shorter or more accurate deterministic product formulas, but also for designing more effective statistical descriptors of randomized trajectories. In continuous TE-PAI, such structure-aware and state-aware descriptors could serve as stratification statistics that separate trajectories according to their expected contribution to sampling fluctuations, thereby improving variance reduction without modifying the underlying unbiased estimator.

Machine learning and adaptive sampling provide natural routes toward a more systematic design of such stratification rules. In a pilot stage, trajectory data could be used to learn features that are predictive of large estimator fluctuations, thereby defining strata adapted to the Hamiltonian, observable, and initial state. The allocation of samples could then be updated adaptively according to the estimated within-stratum variances, in the spirit of optimal allocation in stratified sampling and adaptive stratified Monte Carlo methods. Such an approach would allow the stratification rule and the sample allocation to be refined as data are collected, while preserving the unbiasedness of the underlying randomized estimator.

We also remark that the stratification strategy is not specific to continuous TE-PAI. It can be applied to randomized or quasiprobability-based simulation methods whose samples admit classically computable trajectory features, including randomized product formulas, qDRIFT-type algorithms, circuit-cutting approaches to clustered Hamiltonian simulation~\cite{harrow2025optimal}, and also to PTER, a recent generalisation to TE-PAI \cite{murota2026unbiased}. Thus, stratification provides a general classical variance-reduction layer that does not modify the sampled quantum circuits. Also, a complementary direction is to reduce gate-level cost. Clifford+T synthesis methods such as Ref.~\cite{bothe2026more} could be applied to the \(\Delta\)-angle rotations in TE-PAI circuits, lowering fault-tolerant gate cost independently of the sampling reduction from stratification. Another complementary direction is to use importance sampling to further reduce resource costs \cite{cugini2026resource}.

Overall, the present work demonstrates that even simple stratification
strategies can substantially reduce the sampling cost of continuous TE-PAI. This provides a foundation for more systematic variance-reduction techniques in randomized Hamiltonian simulation, where future improvements may come from structure-aware statistics or adaptive sample allocation.

\begin{acknowledgments}
    C.K. is supported by JSPS KAKENHI, Grant Number JP26K17052, JST ASPIRE Japan, Grant Number JPMJAP2319 and JST PRESTO Japan, Grant Number JPMJPR25F1. F.H. is thankful for continued support from the Oxford Mathematical Institute Scholarship with Jane Street Graduate Scholarship. J.W.D. thanks the Clarendon Fund Scholarship, University of Oxford. For the purpose of Open Access, the author has applied a CC BY public copyright licence to any Author Accepted Manuscript version arising from this submission.
\end{acknowledgments}

\section*{Data Availability}
The simulation code used in this work is available at \url{https://github.com/fredrikhassel/CTEPAI}.

%----------
% Comment out because we want to see title when writing the manuscript
% \bibliographystyle{apsrev4-2}
\bibliography{ref} 

\appendix
\crefalias{section}{appendix}
\section{Weak limit convergence of TE-PAI}\label{app:weak:limit}
\subsection{Proof of \cref{lemma:weak_limit}}
\begin{proof}
For \(\omega=(t_m,k_m,\ell_m)_{m=1}^{M}\in\Omega\), we write
\(\omega_m=(t_m,k_m,\ell_m)\in E:=[0,T]\times[L]\times\{0,1\}\).
The probability density function is given by
\begin{equation*}
P_{\infty } (\omega )=e^{-\Lambda }\prod _{m=1}^{M}\frac{\Lambda | c_{k_{m}}( t_{m})| }{T\overline{\| c\| _{1}}} p_{\Delta }^{1-\ell_{m}} p_{\pi }^{\ell_{m}}.
\end{equation*}
We use the Lévy Continuity Theorem to prove the weak limit theorem by showing the equivalent statement of pointwise convergence of the Laplace functionals with respect to \(P_N(\omega)\) and \(P_\infty(\omega)\). The Laplace functional for a test function $h( t,k,\ell) \geqslant 0$ is defined as
\[L_{\infty} (h)  =\int _{\Omega }\exp\left( -\sum _{m=1}^{|\omega |} h( \omega _{m})\right) P_{\infty } (\omega )d\omega.\]
Conditioning on the event count \(M\), and using the ordered-time density
\(M!\prod_{m=1}^{M}f(\omega_m)\) on \(\Omega_M\), we obtain
\begin{align*}
L_{\infty}(h)
&=
\sum_{M=0}^{\infty}
\int_{\Omega_M}
\exp\left(
-\sum_{m=1}^{M} h(\omega_m)
\right)
P_{\infty}(\omega)\,d\omega \\
&=
\sum_{M=0}^{\infty}
\frac{e^{-\Lambda}\Lambda^M}{M!}
\int_{\Omega_M}
M!\prod_{m=1}^{M}
\left[
e^{-h(\omega_m)}f(\omega_m)
\right]d\omega,
\end{align*}
where \(f(t,k,\ell)=\frac{|c_k(t)|}{T\overline{\|c\|_1}}p_\Delta^{1-\ell}p_\pi^\ell\). Using the symmetry of the ordered-time density and Fubini's theorem, we get
\begin{align*}
L_{\infty}(h)
&=
\sum_{M=0}^{\infty}
\frac{e^{-\Lambda}\Lambda^M}{M!}
\prod_{m=1}^{M}
\int_E e^{-h(\omega_m)} f(\omega_m)\,d\omega_m \\
&=
e^{-\Lambda}
\sum_{M=0}^{\infty}
\frac{\Lambda^M}{M!}
\left(
\int_E e^{-h(x)} f(x)\,dx
\right)^M .
\end{align*}
Applying the Taylor expansion \(e^z=\sum_{M=0}^{\infty}z^M/M!\), we obtain
\begin{align*}
L_{\infty}(h)
&=
e^{-\Lambda}
\exp\left(
\Lambda\int_E e^{-h(x)}f(x)\,dx
\right) \\
&=
\exp\left(
\Lambda\int_E\left(e^{-h(x)}-1\right)f(x)\,dx
\right).
\end{align*}
Here, note that we have used $\int _{E} f( x) dx=1$.

Next, we consider the Laplace functional of the finite-step TE-PAI distribution. Let \(\Omega^{(N)}=\bigsqcup_{M=0}^{NL}\Omega_M^{(N)}\subset\Omega\), where \(\Omega_M^{(N)}:=\{(t_{j_m},k_m,\ell_m)_{m=1}^{M}:(j_1,k_1)<\cdots<(j_M,k_M)\ \text{lexicographically},\ j_m\in[N],\ k_m\in[L],\ \ell_m\in\{0,1\}\}\). Then the Laplace functional can be written as
\begin{align*}
L_{N} (h) & =\int _{\Omega^{( N)}}\exp\left( -\sum _{m=1}^{|\omega |} h( t_{m} ,k_{m} ,\ell_{m})\right) P_{N} (\omega )d\omega \\
\end{align*}
From the definition of the probability $P_{N} (\omega )$ this can be rewritten as:
\begin{align*}
&\prod _{j=1}^{N}\prod _{k=1}^{L}\Bigl( p_{1} (\theta _{kj} )+p_{2} (\theta _{kj} )e^{-h(t_{j} ,k,0)} +p_{3} (\theta _{kj} )e^{-h(t_{j} ,k,1)}\Bigr)
\end{align*}
and then as
$\prod _{j=1}^{N}\prod _{k=1}^{L}\Bigl( 1+x_{kj}\Bigr)$
where 
\[x_{kj} =p_{2} (\theta _{kj} )\left( e^{-h(t_{j} ,k,0)} -1\right) +p_{3} (\theta _{kj} )\left( e^{-h(t_{j} ,k,1)} -1\right).\]
Since \(\theta_{kj}=O(N^{-1})\) and \(p_2(\theta),p_3(\theta)\) vanish linearly at \(\theta=0\), we have \(p_2(\theta_{kj}),p_3(\theta_{kj})=O(N^{-1})\), and hence \(x_{kj}=O(N^{-1})\). Taking the log, we get 

\begin{align*}
\log L_N(h)
&=
\sum_{j=1}^{N}\sum_{k=1}^{L}\log(1+x_{kj}) \\
&=
\sum_{j=1}^{N}\sum_{k=1}^{L}x_{kj}
+O(N^{-1}).
\end{align*}
Using the first-order expansion of the finite-step TE-PAI probabilities and
\(\Lambda f(t,k,\ell)=G_\Delta |c_k(t)|p_\Delta^{1-\ell}p_\pi^\ell\), this becomes
\[
\Lambda
\sum_{\ell=0}^{1}\sum_{j=1}^{N}\sum_{k=1}^{L}
\frac{T}{N}
f(t_j,k,\ell)
\left(e^{-h(t_j,k,\ell)}-1\right)
+O(N^{-1}).
\]
Here we used
\begin{align*}
p_2(\theta_{kj})
&=
G_\Delta |c_k(t_j)|\frac{T}{N}p_\Delta
+O(N^{-2}),\\
p_3(\theta_{kj})
&=
G_\Delta |c_k(t_j)|\frac{T}{N}p_\pi
+O(N^{-2}),\\
1-p_1(\theta_{kj})
&=
G_\Delta |c_k(t_j)|\frac{T}{N}
+O(N^{-2}).
\end{align*}
Taking \(N\to\infty\), we obtain
\begin{align*}
&\lim_{N\to\infty}\log L_N(h)\\
&=
G_\Delta
\sum_{\ell=0}^{1}\sum_{k=1}^{L}
p_\Delta^{1-\ell}p_\pi^\ell
\int_0^T |c_k(t)|
\left(e^{-h(t,k,\ell)}-1\right)dt \\
&=
\Lambda\int_E
\left(e^{-h(x)}-1\right)f(x)\,dx .
\end{align*}
Therefore,
\[
\lim_{N\to\infty}L_N(h)
=
\exp\left[
\Lambda\int_E
\left(e^{-h(x)}-1\right)f(x)\,dx
\right]
=
L_\infty(h).
\]
This completes the proof.
\end{proof}

\subsection{Proof of \cref{thm:unbiased}}
In this subsection, we prove the unbiasedness theorem, \cref{thm:unbiased}, using \cref{lemma:weak_limit}.

\begin{proof}
Let \(\|\cdot\|\) denote a fixed superoperator norm for which
\(\|\mathcal U_\omega\|\leq 1\). We first show that
\[
\lim_{N\to\infty}
\mathbb E_{P_N}\left[g_\omega^{(N)}\mathcal U_\omega\right]
=
\mathbb E_{P_\infty}\left[g_\omega\mathcal U_\omega\right].
\]
Indeed,
\begin{align*}
&\left\|
\mathbb E_{P_N}\left[g_\omega^{(N)}\mathcal U_\omega\right]
-
\mathbb E_{P_\infty}\left[g_\omega\mathcal U_\omega\right]
\right\| \\
&\leq
\left\|
\mathbb E_{P_N}\left[g_\omega^{(N)}\mathcal U_\omega\right]
-
\mathbb E_{P_N}\left[g_\omega\mathcal U_\omega\right]
\right\| \\
&\quad+
\left\|
\mathbb E_{P_N}\left[g_\omega\mathcal U_\omega\right]
-
\mathbb E_{P_\infty}\left[g_\omega\mathcal U_\omega\right]
\right\| \\
&=: E_N+F_N .
\end{align*}
The first term satisfies
\[
E_N
\leq
\mathbb E_{P_N}\left[
\left|g_\omega^{(N)}-g_\omega\right|
\|\mathcal U_\omega\|
\right]
\leq
\mathbb E_{P_N}\left[
\left|g_\omega^{(N)}-g_\omega\right|
\right].
\]
Since \(g_\omega^{(N)}=g_\omega+O(N^{-1})\) uniformly on the finite-step trajectory space, we have \(E_N\to0\).

For the second term, \cref{lemma:weak_limit} gives \(P_N\Rightarrow P_\infty\). Since the map $\omega\mapsto g_\omega\mathcal U_\omega$ is bounded and continuous with respect to the trajectory topology used in \cref{lemma:weak_limit}, it follows that \(F_N\to0\). Hence
\[
\lim_{N\to\infty}
\mathbb E_{P_N}\left[g_\omega^{(N)}\mathcal U_\omega\right]
=
\mathbb E_{P_\infty}\left[g_\omega\mathcal U_\omega\right].
\]
By the finite-step TE-PAI identity,
\[
\mathbb E_{P_N}\left[g_\omega^{(N)}\mathcal U_\omega\right]
=
\mathcal U_N .
\]
Taking \(N\to\infty\) and using \(\mathcal U_N\to\mathcal U\), we obtain
\[
\mathbb E_{P_\infty}\left[g_\omega\mathcal U_\omega\right]
=
\lim_{N\to\infty}
\mathbb E_{P_N}\left[g_\omega^{(N)}\mathcal U_\omega\right]
=
\lim_{N\to\infty}\mathcal U_N
=
\mathcal U .
\]
This proves the unbiasedness of the continuous TE-PAI estimator.
\end{proof}
\section{Continuous TE-PAI from Dyson series}
\label{app:dyson_derivation_continuous_tepai}
We now derive the continuous TE-PAI estimator from the Dyson expansion. This derivation clarifies the role of the Poisson rate, the signed \(\pi\)-rotation events, and the exponential weight in \eqref{eq:weight}. Let \(\mathcal U(t)\) denote the exact time-evolution channel. Define the Liouvillian
\begin{equation}\label{eq:Liouvillian}
\mathcal L(t)(\rho)
:=
-i[H(t),\rho].
\end{equation}
Since $ H(t)=\sum_{k=1}^L c_k(t)P_k,$
we have
\[
\mathcal L(t)=\sum_{k=1}^L c_k(t)\mathcal L_k,
\qquad
\mathcal L_k(\rho):=-i[P_k,\rho].
\]
The exact channel is
\[
\mathcal U(T)
=
\mathcal T\exp\left(
\int_0^T \mathcal L(t)\,dt
\right).
\]
Its Dyson expansion is
\[
\mathcal U(T)
=
\sum_{M=0}^{\infty}
\int_{0\le t_1\le\cdots\le t_M\le T}
dt_1\cdots dt_M\,
\mathcal L(t_M)\cdots\mathcal L(t_1).
\]
Substituting Eq. (\cref{eq:Liouvillian}) gives
\begin{align*}
\mathcal U(T)
&=
\sum_{M=0}^{\infty}
\sum_{k_1,\ldots,k_M}
\int_{0\le t_1\le\cdots\le t_M\le T}
dt_1\cdots dt_M\\
&\qquad\times
\left(
\prod_{m=1}^M c_{k_m}(t_m)
\right)
\mathcal L_{k_M}\cdots\mathcal L_{k_1}.
\end{align*}
Thus a Dyson trajectory is an ordered sequence of generator insertions \((t_1,k_1),\ldots,(t_M,k_M)\), with ordered product \(\mathcal L_{k_M}\cdots\mathcal L_{k_1}\). The generator \(\mathcal L_k\) is not itself a physical quantum channel, but its flow is: for a Pauli \(P\), \(e^{\theta\mathcal L_P/2}=\mathcal R_P(\theta)\), where \(\mathcal R_P(\theta)(\rho)=e^{-i\theta P/2}\rho e^{i\theta P/2}\). Since \(P^2=I\), the Taylor expansion of the superoperator flow \(e^{\theta\mathcal L_P/2}\) closes exactly, giving
\[
e^{\theta\mathcal L_P/2}
=
\mathcal R_P(\theta)
=
\cos^2\frac{\theta}{2}\mathcal I
+
\sin^2\frac{\theta}{2}\mathcal R_P(\pi)
+
\frac{\sin\theta}{2}\mathcal L_P.
\]

The identity above can first be rearranged to isolate the generator. For any Pauli \(P\),
\begin{align*}
\pm \mathcal{L}_{P}
&=
\frac{2}{\sin \Delta}
\left[
\mathcal{R}_{P}(\pm \Delta)
-\sin^{2}\frac{\Delta}{2}\mathcal{R}_{P}(\pi)
-\cos^{2}\frac{\Delta}{2}\mathcal{I}
\right] \\
&=
-\cot\frac{\Delta}{2}\,\mathcal{I}
+
\frac{3-\cos \Delta}{\sin \Delta}
\left[
p_{\Delta}\mathcal{R}_{P}(\pm \Delta)
-
p_{\pi}\mathcal{R}_{P}(\pi)
\right].
\end{align*}

Summing over \(k\) gives the decomposition of the full Liouvillian
\[
\mathcal L(t)
=
-\|c(t)\|_1\cot\frac{\Delta}{2}\mathcal I
+
B(t),
\]
where
\[
B(t)
:=
G_\Delta
\sum_{k=1}^L |c_k(t)|
\left[
p_\Delta \mathcal R_{P_k}(\operatorname{sgn}(c_k(t))\Delta)
-
p_\pi \mathcal R_{P_k}(\pi)
\right].
\]
Hence
\[
\mathcal U(T)
=
\exp\left[
-\overline{\|c\|_1}T
\cot\frac{\Delta}{2}
\right]
\mathcal T\exp\left[
\int_0^T B(t)dt
\right].
\]
Let
$
\omega=(t_m,k_m,\ell_m)_{m=1}^M, 0\le t_1\le\cdots\le t_M\le T,
$
where \(\ell_m=0\) denotes a \(\Delta\)-rotation and \(\ell_m=1\) denotes a \(\pi\)-rotation. Define
\[
\mathcal R_{k,t,0}:=\mathcal R_{P_k}(\operatorname{sgn}(c_k(t))\Delta),
\qquad
\mathcal R_{k,t,1}:=\mathcal R_{P_k}(\pi),
\]
and
\[
\mathcal U_\omega
:=
\mathcal R_{k_M,t_M,\ell_M}\circ\cdots\circ
\mathcal R_{k_1,t_1,\ell_1}.
\]
The positive coefficient measure induced by the Dyson expansion of \(B(t)\) is
\[
d\mu(\omega)
:=
\prod_{m=1}^M
G_\Delta |c_{k_m}(t_m)|
p_\Delta^{1-\ell_m}
p_\pi^{\ell_m}\,dt_m .
\]
Thus
\[
\mathcal T e^{\int_0^T B(t)dt}
=
\int_{\Omega}
d\mu(\omega)\,
\Big(\prod_{m=1}^M(-1)^{\ell_m}\Big)
\mathcal U_\omega .
\]
Here \(\Omega=\bigsqcup_{M\ge0}\Omega_M\) is the space of ordered marked trajectories
\(\omega=(t_m,k_m,\ell_m)_{m=1}^M\).
The measure \(d\mu\) has total mass
\[
\mu(\Omega)=e^\Lambda,
\qquad
\Lambda=G_\Delta\,\overline{\|c\|_1}T.
\]
We therefore define the normalized trajectory measure
\[
dP_\infty(\omega):=e^{-\Lambda}d\mu(\omega).
\]
Then
\[
\mathcal T e^{\int_0^T B(t)dt}
=
e^\Lambda
\int_\Omega
dP_\infty(\omega)\,
\Big(\prod_{m=1}^M(-1)^{\ell_m}\Big)
\mathcal U_\omega .
\]
Combining this with the scalar prefactor gives
\[
\mathcal U(T)
=
\int_\Omega
dP_\infty(\omega)\,
g_\omega\,\mathcal U_\omega,
\]
where we have used
\[
\Lambda-\overline{\|c\|_1}T\cot\frac{\Delta}{2}=2\overline{\|c\|_1}T\tan\frac{\Delta}{2}.
\]

The normalized measure \(P_\infty\) is precisely a marked Poisson law. Indeed, using $\Lambda=G_\Delta T\overline{\|c\|_1},$
this can be rewritten as
\[
dP_\infty(\omega)
=
e^{-\Lambda}\frac{\Lambda^M}{M!}\,
M!
\prod_{m=1}^M
\left[
\frac{|c_{k_m}(t_m)|}{T\overline{\|c\|_1}}
p_\Delta^{1-\ell_m}
p_\pi^{\ell_m}
dt_m
\right].
\]
Thus \(M\sim \mathrm{Poisson}(\Lambda)\). Conditional on \(M\), the ordered event times are obtained by drawing \(M\) i.i.d. times with density
\[
\frac{\|c(t)\|_1}{T\overline{\|c\|_1}}
\]
and sorting them. Given an event time \(t_m\), the Pauli index and event type are drawn according to
\begin{align*}
&\Pr(k_m=k\mid t_m)
=
\frac{|c_k(t_m)|}{\|c(t_m)\|_1},\\
&\Pr(\ell_m=0)
=
p_\Delta,\quad
\Pr(\ell_m=1)
=
p_\pi.
\end{align*}
Combining the normalized trajectory law with the scalar prefactor, we obtain
\[
\mathcal U(T)
=
\int_\Omega dP_\infty(\omega)\,
g_\omega\,\mathcal U_\omega.
\]
Therefore,
\[
\mathbb E_{\omega\sim P_\infty}
\left[
g_\omega\,\mathcal U_\omega
\right]
=
\mathcal U(T).
\]
This proves the unbiasedness of the continuous TE-PAI estimator.

\section{Variance Theory}
\label{app:main_theory}
\subsection{Commutation-class statistics}

We first identify an ideal statistic with zero residual risk. The idea is to group together all trajectories that realize the same channel because they differ only by adjacent swaps of channel-commuting letters. Within each such stratum, the channel value is constant, and hence the within-stratum variance is zero.

\begin{definition}[Local swap defect]
\label{def:local_swap_defect}
For two letters \(a,b\in\mathcal A\), define the local swap defect
\[
\Delta_{ab}
:=
\left\|
\Gamma_a\circ\Gamma_b-\Gamma_b\circ\Gamma_a
\right\|_{HS}.
\]
\end{definition}

The defect \(\Delta_{ab}\) measures the channel-level effect of exchanging the order of two adjacent letters. 
It vanishes precisely when the corresponding single-letter channels commute under composition:
\[
\Delta_{ab}=0
\quad\Longleftrightarrow\quad
\Gamma_a\circ\Gamma_b=\Gamma_b\circ\Gamma_a.
\]
It follows that
\[
\Delta_{aa}=0.
\]

For unitary trajectories, the local defect has an exact expression in terms of a group commutator.

\begin{lemma}[Exact local channel distance for unitary trajectories]
\label{lemma:exact_local_channel_distance}
Suppose
\[
\Gamma_a(\rho)=U_a\rho U_a^\dagger,
\]
where \(U_a\) is unitary for each \(a\in\mathcal A\). Then, for any
\(a,b\in\mathcal A\),
\[
\Delta_{ab}^2
=
2\left(
d^2-
\left|\Tr(U_a^\dagger U_b^\dagger U_aU_b)\right|^2
\right).
\]
In particular, the expression is independent of the phases chosen for the implementing unitaries, and is bounded by $0\le \Delta_{ab}^2 \le 2d^2.$
\end{lemma}

Define the set of locally channel-invisible pairs
\[
I_0
:=
\left\{
\{a,b\}\subseteq\mathcal A:
a\neq b,\ \Delta_{ab}=0
\right\}.
\]
So \(I_0\) consists of the distinct letter pairs whose unitary channels commute.

We say that two words \(\underline x,\underline y\in\mathcal A^m\) are \(I_0\)-equivalent if one can be transformed into the other by a finite sequence of adjacent swaps
\[
ab\leftrightarrow ba
\]
using only pairs \(\{a,b\}\in I_0\), and we write
\[
\underline x\sim_{I_0}\underline y
\]
and denote the equivalence class of \(\underline x\) by
\[
[\underline x]_{I_0}.
\]
The associated commutation-class statistic is
\[
T_{I_0}(\underline X):=[\underline X]_{I_0}.
\]
Equivalently, \(T_{I_0}\) is the quotient of trajectory space generated by all channel-invisible adjacent commutations. 
In formal-language terminology, this is the trace-monoid quotient associated with the independence relation \(I_0\).

\begin{proposition}[Commutation classes have zero residual risk]
\label{prop:commutation_class_statistic_exact}
Conditioning on \(T_{I_0}\) leaves no residual risk:
\[
\mathcal R(\Gamma(\underline X),T_{I_0})=0.
\]
\end{proposition}

The full trajectory statistic
\[
T_{\mathrm{traj}}(\underline X):=\underline X
\]
also has zero residual risk, but only trivially. The statistic \(T_{I_0}\) is more informative: it discards all order information that is invisible to the realized channel. If the channel-level variance is zero then it follows that the variance of all downstream tasks is also zero (since it channel level variance strictly dominates).

\subsubsection{Conditional residual risk and commutation-aware coarsenings}

The Hilbert-space projection identity gives the variance decomposition
\[
\mathcal R(\widehat\Phi)
=
\mathbb E\left\|
\mathbb E[\widehat\Phi\mid S]-\overline\Phi
\right\|_{HS}^{2}
+
\mathcal R(\widehat\Phi,S).
\]
Thus refining \(S\) can only decrease the residual risk. In particular, if \(S_1\) is finer than \(S_2\), written \(S_1\succeq S_2\), then
\[
\mathcal R(\widehat\Phi,S_2)
=
\mathcal R(\widehat\Phi,S_1)
+
\mathbb E\left\|
\mathbb E[\widehat\Phi\mid S_1]
-
\mathbb E[\widehat\Phi\mid S_2]
\right\|_{HS}^{2}.
\]
This identity is the basic mechanism behind the residual hierarchy below.

The local swap defect $\Delta_{ab}$ vanishes precisely when the two single-letter channels commute under composition.
Define the independence relation 
\[
I_0
:=
\left\{
\{a,b\}\subseteq\mathcal A:
a\neq b,\ \Delta_{ab}=0
\right\}.
\]
Two words \(\underline x,\underline y\in\mathcal A^m\) are called
\(I_0\)-equivalent if one can be transformed into the other by adjacent swaps
\[
ab\leftrightarrow ba
\]
using only pairs \(\{a,b\}\in I_0\). We also say that the two letters are channel invisible, and write
\[
\underline x\sim_{I_0}\underline y,
\qquad
[\underline x]_{I_0}
\]
for the corresponding equivalence relation and equivalence class. Then, the commutation-class statistic \(T_{I_0}\) is the trace-monoid quotient obtained by identifying all adjacent swaps that are invisible at the channel level.

Since every allowed adjacent swap inside an \(I_0\)-class leaves the realized
channel unchanged, \(\Gamma(\underline x)\) is constant on each
\(I_0\)-class. Hence
\[
\mathcal R(\widehat\Phi,T_{I_0})=0.
\]
The full trajectory statistic \(T_{\mathrm{traj}}(\underline X)=\underline X\)
also has zero residual risk, but only trivially. 
The statistic \(T_{I_0}\) is more informative: it removes only the ordering information that is invisible to the realized channel.

\paragraph{Thresholded order loss.}

Although \(T_{I_0}\) has zero residual risk, it is typically too fine to be operationally useful. A natural family of coarser statistics is obtained by remembering the relative order only of letter pairs whose local swap defect is larger than a threshold.

For \(\tau\ge0\), define
\[
E_\tau
:=
\left\{
\{a,b\}\subseteq\mathcal A:
a\neq b,\ \Delta_{ab}>\tau
\right\},
\quad
I_\tau:=\binom{\mathcal A}{2}\setminus E_\tau .
\]
Let
\[
S_\tau(\underline X):=[\underline X]_{I_\tau}
\]
be the quotient statistic generated by adjacent swaps of all pairs with \(\Delta_{ab}\le\tau\). Thus \(S_\tau\) remembers the relative order of pairs with large local defect and forgets the relative order of pairs with small local defect.

At \(\tau=0\), this recovers the exact commutation-class statistic,
\(
S_0=T_{I_0}.
\)
At the opposite endpoint, let
\[
\tau_{\max}:=\max_{a\neq b}\Delta_{ab}.
\]
Then \(E_{\tau_{\max}}=\varnothing\), so \(S_{\tau_{\max}}\) allows adjacent swaps of every distinct pair of letters. 
Hence \(S_{\tau_{\max}}\) forgets all ordering information and retains only the multiset of letters.

Let
\[
\mathbf N(\underline x)
=
\bigl(N_a(\underline x)\bigr)_{a\in\mathcal A}
\]
be the counts vector, also known as the Parikh vector. On each fixed-length word space \(\mathcal A^m\), the statistic \(S_{\tau_{\max}}\) is equivalent to \(\mathbf N\).

Since \(\mathcal A\) is finite, there are finitely many distinct positive defect values. Let
\[
0=\tau_0<\tau_1<\cdots<\tau_K=\tau_{\max}
\]
be the corresponding threshold levels, and abbreviate
\[
S_j:=S_{\tau_j}.
\]
Then the thresholded statistics form the finite refinement chain
\[
T_{I_0}=S_0
\succeq
S_1
\succeq
\cdots
\succeq
S_K=\mathbf N.
\]
Iterating the projection identity gives the exact additive decomposition
\[
\mathcal R(\widehat\Phi,\mathbf N)
=
\sum_{j=0}^{K-1}
\eta_j^{\mathrm{ord}},
\]
where
\[
\eta_j^{\mathrm{ord}}
:=
\mathbb E\left\|
\mathbb E[\widehat\Phi\mid S_j]
-
\mathbb E[\widehat\Phi\mid S_{j+1}]
\right\|_{HS}^{2}.
\]
Thus \(\mathcal R(\widehat\Phi,\mathbf N)\) is the total residual risk incurred by moving from the exact commutation-class statistic \(T_{I_0}\) down to the fully order-forgetting counts vector \(\mathbf N\).

One can coarsen further beyond \(\mathbf N\), for example by grouping letters and retaining only grouped counts. Such coarsenings introduce additional identity-loss or label-loss terms. 
The Green-kernel formula below isolates the order-loss core at the counts level.

\subsubsection{Exact Green-kernel formula for the order-loss core}
\label{sec:green_kernel_order_loss}

We now give an exact expression for the counts-level residual
\[
\mathcal R(\widehat\Phi,\mathbf N)
=
\mathbb E\!\left[\Var(\Gamma(\underline X)\mid \mathbf N)\right].
\]
The expression is a Green-kernel quadratic form on adjacent-swap graphs. Its
edge variables are context-dependent commutators of channels. 
A looser, but numerically simpler, norm-only Poincar\'e--Dirichlet bound is then recovered by replacing the exact Green kernel by a worst-case spectral bound and discarding the inner-product geometry between different commutator increments.

\paragraph{Weighted graph notation.}

Let \(V\) be a finite set equipped with a probability law \(\pi\), with \(\pi(x)>0\) for all \(x\in V\). 
Let
\[
c:V\times V\to[0,\infty)
\]
be a symmetric edge-weight function with \(c(x,x)=0\). 
We write \(x\sim y\) when \(c(x,y)>0\), and assume that the graph is connected. The specific choice of edge weights $c$ is not important, since it is normalised by the spectral decomposition.

The associated weighted graph Laplacian acts on scalar functions \(f:V\to\mathbb C\) by
\[
(L_cf)(x)
:=
\frac1{\pi(x)}
\sum_{y\in V}
c(x,y)\bigl(f(x)-f(y)\bigr).
\]
It is self-adjoint and positive semidefinite on \(L^2(\pi)\), with
\[
\langle f,L_cf\rangle_\pi
=
\frac12
\sum_{x,y\in V}
c(x,y)|f(x)-f(y)|^2.
\]

The same notation applies to Hilbert-space-valued functions \(F:V\to\mathsf H\). In our application, \(\mathsf H\) is the Hilbert space of superoperators equipped with the Hilbert--Schmidt inner product. The variance is
\[
\Var_\pi(F)
:=
\sum_{x\in V}
\pi(x)
\left\|
F(x)-\mathbb E_\pi F
\right\|_{\mathsf H}^{2}.
\]

Choose an arbitrary orientation for every undirected edge and let \(E^+\) denote the resulting oriented edge set. 
For \(e=(x,y)\in E^+\), define
\[
\nabla F(e):=F(y)-F(x).
\]
The edge inner product is
\[
\langle A,B\rangle_E
:=
\sum_{e\in E^+}
c_e\,
\langle A(e),B(e)\rangle_{\mathsf H},
\qquad
c_e:=c(x,y).
\]
With this convention, the Laplacian acts as
\[
\mathcal E_c(F,F)
:=
\frac12
\sum_{x,y\in V}
c(x,y)\|F(x)-F(y)\|_{\mathsf H}^{2}
=
\|\nabla F\|_E^2.
\]
The adjoint \(\nabla^\ast\) is taken with respect to this edge inner product and
the vertex inner product on \(L^2(\pi;\mathsf H)\), and
\[
L_c=\nabla^\ast\nabla.
\]

Let \(L_c^+\) be the Moore--Penrose pseudoinverse of \(L_c\), acting as the
inverse of \(L_c\) on the mean-zero subspace and as zero on constants. Define
the edge-space Green operator
\[
K_c^{\mathrm{var}}
:=
\nabla (L_c^+)^2\nabla^\ast .
\]

\begin{lemma}[Green-kernel variance identity]
\label{lem:green_kernel_variance_identity}
Let \((V,\pi,c)\) be a finite connected weighted graph, and let
\(F:V\to\mathsf H\) be Hilbert-space-valued. Then
\[
\Var_\pi(F)
=
\left\langle
\nabla F,
K_c^{\mathrm{var}}\nabla F
\right\rangle_E .
\]
Equivalently, if
\[
0=\lambda_0<\lambda_1\le\lambda_2\le\cdots
\]
are the eigenvalues of \(L_c\), with a real \(L^2(\pi)\)-orthonormal eigenbasis
\(\{\psi_r\}_{r\ge0}\), then
\[
\Var_\pi(F)
=
\sum_{r\ge1}\|C_r\|_{\mathsf H}^2,
\qquad
C_r:=\mathbb E_\pi[F\psi_r],
\]
and also
\[
\Var_\pi(F)
=
\sum_{r\ge1}
\frac1{\lambda_r^2}
\left\|
\sum_{e\in E^+}
c_e\,\nabla\psi_r(e)\,\nabla F(e)
\right\|_{\mathsf H}^{2}.
\]
\end{lemma}

\begin{proof}
Subtract the mean and write
\[
F_0:=F-\mathbb E_\pi F.
\]
Then \(F_0\) lies in the mean-zero subspace, so
\[
F_0
=
L_c^+L_cF_0
=
L_c^+\nabla^\ast\nabla F.
\]
Expanding in a real orthonormal eigenbasis of \(L_c\),
\[
F_0=\sum_{r\ge1}C_r\psi_r,
\qquad
C_r=\mathbb E_\pi[F\psi_r],
\]
gives
\[
\Var_\pi(F)=\sum_{r\ge1}\|C_r\|_{\mathsf H}^2.
\]
Since \(L_c\psi_r=\lambda_r\psi_r\), integration by parts yields
\begin{align*}
C_r
=
\frac1{\lambda_r}
\mathbb E_\pi[F\,L_c\psi_r]
&=
\frac1{\lambda_r}
\langle\nabla F,\nabla\psi_r\rangle_E\\
&=
\frac1{\lambda_r}
\sum_{e\in E^+}
c_e\,\nabla\psi_r(e)\,\nabla F(e).
\end{align*}
Substitution gives the spectral formula. The operator identity follows by
rewriting the same expression as
\[
\Var_\pi(F)
=
\left\langle
\nabla F,
\nabla(L_c^+)^2\nabla^\ast\nabla F
\right\rangle_E
=
\left\langle
\nabla F,
K_c^{\mathrm{var}}\nabla F
\right\rangle_E .
\]
\end{proof}

We now apply this identity to the count sectors of the trajectory space. For a
count vector \(\mathbf n\) with \(\Pr(\mathbf N=\mathbf n)>0\), define
\[
\Omega_{\mathbf n}
:=
\{\underline x\in\mathcal W:\mathbf N(\underline x)=\mathbf n\},
\qquad
\mu_{\mathbf n}
:=
\operatorname{Law}(\underline X\mid \mathbf N=\mathbf n).
\]
Choose a symmetric edge-weight function \(c_{\mathbf n}\) on
\(\operatorname{supp}\mu_{\mathbf n}\) such that
\[
c_{\mathbf n}(\underline x,\underline y)>0,
\quad
\underline x\neq\underline y,
\]
if \(\underline y\) is obtained from \(\underline x\) by one adjacent transposition. Note that these weights define a reversible adjacent-swap Markov chain with invariant law \(\mu_{\mathbf n}\), with jump rates
\[
q_{\mathbf n}(\underline x,\underline y)
=
\frac{
c_{\mathbf n}(\underline x,\underline y)
}{
\mu_{\mathbf n}(\underline x)
}.
\]
The exact Green-kernel identity below holds for every such admissible conductance.
Different choices of \(c_{\mathbf n}\) give different Green-kernel representations of the same sector variance.
Note also that the resulting adjacent-swap graph on \(\operatorname{supp}\mu_{\mathbf n}\) is connected. 
This is automatic since \(\mu_{\mathbf n}\) has full support on \(\Omega_{\mathbf n}\). A positive weight is assigned to every allowed adjacent swap, since any two words with the same multiplicities are connected by adjacent transpositions.

Let
\[
L_{\mathbf n},\qquad
L_{\mathbf n}^+,\qquad
\nabla_{\mathbf n},\qquad
K_{\mathbf n}^{\mathrm{var}}
:=
\nabla_{\mathbf n}(L_{\mathbf n}^+)^2\nabla_{\mathbf n}^\ast
\]
denote the corresponding Laplacian, pseudoinverse, gradient, and edge-space Green operator respectively.

For an oriented adjacent-swap edge
\[
e=(\underline x,\underline y)\in E_{\mathbf n}^+,
\]
define the channel increment
\[
D_e\Gamma
:=
\Gamma(\underline y)-\Gamma(\underline x).
\]
If \(\underline x\) and \(\underline y\) differ by swapping an adjacent pair
\(a,b\), we write
\[
\tau(e):=\{a,b\}
\]
for the unordered edge type.

\begin{proposition}[Exact Green-kernel formula for the counts residual]
\label{prop:exact_green_kernel_counts_residual}
For each count sector \(\Omega_{\mathbf n}\), define the sectorwise residual
\[
\mathcal R_{\mathbf n}
:=
\Var_{\mu_{\mathbf n}}(\Gamma)
=
\mathbb E_{\mu_{\mathbf n}}
\left\|
\Gamma(\underline X)
-
\mathbb E_{\mu_{\mathbf n}}\Gamma(\underline X)
\right\|_{HS}^{2}.
\]
Then
\[
\mathcal R_{\mathbf n}
=
\sum_{e,e'\in E_{\mathbf n}^+}
c_e c_{e'}\,
\mathcal K_{\mathbf n}^{\mathrm{var}}(e,e')\,
\left\langle
D_e\Gamma,D_{e'}\Gamma
\right\rangle_{HS}.
\]
Here \(\mathcal K_{\mathbf n}^{\mathrm{var}}(e,e')\) is the scalar kernel of
\(K_{\mathbf n}^{\mathrm{var}}\), defined by
\[
(K_{\mathbf n}^{\mathrm{var}}A)(e)
=
\sum_{e'\in E_{\mathbf n}^+}
c_{e'}\,
\mathcal K_{\mathbf n}^{\mathrm{var}}(e,e')A(e').
\]

Equivalently, if
\[
L_{\mathbf n}\psi_r^{(\mathbf n)}
=
\lambda_r(\mathbf n)\psi_r^{(\mathbf n)},
\qquad
r\ge1,
\]
is an \(L^2(\mu_{\mathbf n})\)-orthonormal nonconstant eigenbasis, then
\[
\mathcal R_{\mathbf n}
=
\sum_{r\ge1}
\frac1{\lambda_r(\mathbf n)^2}
\left\|
\sum_{e\in E_{\mathbf n}^+}
c_e\,
\nabla\psi_r^{(\mathbf n)}(e)\,
D_e\Gamma
\right\|_{HS}^{2}.
\]
The global counts residual is obtained by averaging over count sectors:
\[
\mathcal R(\widehat\Phi,\mathbf N)
=
\mathbb E_{\mathbf N}\bigl[\mathcal R_{\mathbf N}\bigr].
\]
\end{proposition}

\begin{proof}
Apply \cref{lem:green_kernel_variance_identity} to the finite weighted graph
\[
(\operatorname{supp}\mu_{\mathbf n},\mu_{\mathbf n},c_{\mathbf n})
\]
and to the Hilbert-space-valued function
\[
F(\underline x):=\Gamma(\underline x).
\]
This gives
\[
\Var_{\mu_{\mathbf n}}(\Gamma)
=
\left\langle
\nabla_{\mathbf n}\Gamma,
K_{\mathbf n}^{\mathrm{var}}\nabla_{\mathbf n}\Gamma
\right\rangle_E .
\]
Writing this edge inner product in coordinates gives the Green-kernel formula.
The spectral form follows from the eigenfunction representation in
\cref{lem:green_kernel_variance_identity}. Finally,
\[
\mathcal R(\widehat\Phi,\mathbf N)
=
\mathbb E\!\left[
\Var(\Gamma(\underline X)\mid\mathbf N)
\right]
=
\mathbb E_{\mathbf N}[\mathcal R_{\mathbf N}],
\]
as claimed.
\end{proof}

\paragraph{Contextual commutator interpretation.}

The edge increments in
\cref{prop:exact_green_kernel_counts_residual} have a direct commutator
interpretation. Suppose an oriented edge \(e=(\underline x,\underline y)\)
swaps an adjacent pair \(a,b\). Write
\[
\underline x=uabv,
\qquad
\underline y=ubav,
\]
where \(u\) and \(v\) are the prefix and suffix surrounding the swapped pair.
With the convention
\[
\Gamma(x_1,\ldots,x_m)
=
\Gamma_{x_m}\circ\cdots\circ\Gamma_{x_1},
\]
we have
\[
D_e\Gamma
=
\Gamma_v\circ
\bigl(
\Gamma_a\circ\Gamma_b-\Gamma_b\circ\Gamma_a
\bigr)
\circ
\Gamma_u.
\]
Define the local channel commutator increment
\[
C_{ab}
:=
\Gamma_a\circ\Gamma_b-\Gamma_b\circ\Gamma_a.
\]
Then
\[
D_e\Gamma
=
\Gamma_v\circ C_{ab}\circ\Gamma_u.
\]
Because \(\Gamma_u\) and \(\Gamma_v\) are unitary channels, left and right composition by them are Hilbert--Schmidt isometries on the superoperator space.
Therefore
\[
\|D_e\Gamma\|_{HS}
=
\|C_{ab}\|_{HS}
=
\Delta_{ab}.
\]
The exact residual is therefore a Green-kernel quadratic form in local commutators transported by their trajectory contexts.

Crucially, the exact formula depends not only on the magnitudes \(\Delta_{ab}\), but also on the Hilbert--Schmidt inner products of the contextual commutators
\[
\left\langle
D_e\Gamma,D_{e'}\Gamma
\right\rangle_{HS}.
\]
Equivalently,
\[
\left\langle
D_e\Gamma,D_{e'}\Gamma
\right\rangle_{HS}
=
\left\langle
\Gamma_{v_e}\circ C_{\tau(e)}\circ\Gamma_{u_e},
\Gamma_{v_{e'}}\circ C_{\tau(e')}\circ\Gamma_{u_{e'}}
\right\rangle_{HS}.
\]
Using unitary invariance, this can be rewritten as a context-relative commutator correlation:
\[
\left\langle
D_e\Gamma,D_{e'}\Gamma
\right\rangle_{HS}
=
\left\langle
C_{\tau(e)},
\Gamma_{v_e}^{-1}\circ\Gamma_{v_{e'}}
\circ
C_{\tau(e')}
\circ
\Gamma_{u_{e'}}\circ\Gamma_{u_e}^{-1}
\right\rangle_{HS}.
\]
Thus the counts residual is determined by the full Hilbert--Schmidt geometry of contextualized commutator increments across the adjacent-swap graph. 

\begin{corollary}[Zero counts residual in the commuting case]
\label{cor:green_kernel_zero_commuting_counts}
Suppose that, on a count sector \(\Omega_{\mathbf n}\), every adjacent-swap edge
\(e\) satisfies
\[
D_e\Gamma=0.
\]
Then
\[
\mathcal R_{\mathbf n}=0.
\]
Consequently, if all local channels commute pairwise on every sector with nonzero probability, then
\[
\mathcal R(\widehat\Phi,\mathbf N)=0.
\]
\end{corollary}

This means that the counts vector has zero residual in the fully commuting case, this propagates down also into scalar observables. 
\paragraph{Remark on signed estimators.}

The above results are written for the unitary trajectory estimator
\(\widehat\Phi=\Gamma(\underline X)\) (like qDRIFT). For signed or quasiprobability estimators, such as TE-PAI, one instead applies the same Green-kernel identity to
\[
F(\underline X):=g(\underline X)\Gamma(\underline X).
\]
All formulas remain valid with \(D_e\Gamma\) replaced by
\[
D_eF
:=
F(\underline y)-F(\underline x).
\]
If the chosen count statistic determines the weight \(g\), then \(g\) is constant inside each count sector and
\[
D_eF
=
g(\mathbf n)\,D_e\Gamma,
\qquad
\mathcal R_{\mathbf n}(F)
=
g(\mathbf n)^2\,\mathcal R_{\mathbf n}(\Gamma).
\]
This is the relevant situation when the alphabet includes the TE-PAI gate type label, so that the number of signed \(\pi\)-events is count-measurable. Similarly, the same Green-kernel result can be invoked at the scalar level by replacing $F=\Gamma$ with the scalar trajectory after mapping into a scalar value through a linear functional $L$.

\subsubsection{Poincar\'e--Dirichlet bound as a norm corollary}
\label{sec:poincare_order_loss}

The exact Green-kernel formula immediately yields a simpler Poincar\'e--Dirichlet estimate as a coarser upper bound. These are the results quoted in \cref{sec:structure_aware_variance_reduction} of the main text. 
This bound keeps only the norms of the contextual commutators and the spectral gap of the adjacent-swap graph, discarding the inner-product correlations between different edge increments.

For each count vector \(\mathbf n\), let
\[
\lambda(\mathbf n):=\lambda_1(L_{\mathbf n})
\]
be the smallest positive eigenvalue of the sector Laplacian. 
For each unordered pair \(a<b\), define the total edge weight of adjacent swaps of type \(\{a,b\}\) by
\[
w_{ab}(\mathbf n)
:=
\sum_{\substack{
e\in E_{\mathbf n}^+\\
\tau(e)=\{a,b\}
}}
c_e.
\]
Equivalently, in unoriented-edge notation,
\[
w_{ab}(\mathbf n)
=
\frac12
\sum_{\substack{
\underline x,\underline y\in\operatorname{supp}\mu_{\mathbf n}\\
\underline x,\underline y
\text{ differ by an adjacent swap of }a\text{ and }b
}}
c_{\mathbf n}(\underline x,\underline y).
\]

\begin{proposition}[Weighted-graph bound for the counts residual]
\label{prop:poincare_counts_residual_law_agnostic}
For every count sector with connected adjacent-swap graph,
\[
\mathcal R_{\mathbf n}
\le
\frac1{\lambda(\mathbf n)}
\sum_{a<b}
w_{ab}(\mathbf n)\Delta_{ab}^{2}.
\]
Consequently,
\[
\mathcal R(\widehat\Phi,\mathbf N)
\le
\sum_{a<b}
\omega_{ab}\Delta_{ab}^{2},
\qquad
\omega_{ab}
:=
\mathbb E_{\mathbf N}\left[
\frac1{\lambda(\mathbf N)}
w_{ab}(\mathbf N)
\right].
\]
\end{proposition}

\begin{proof}
The Green-kernel identity and the spectral gap bound give
\[
\mathcal R_{\mathbf n}
=
\Var_{\mu_{\mathbf n}}(\Gamma)
\le
\frac1{\lambda(\mathbf n)}
\|\nabla_{\mathbf n}\Gamma\|_E^2.
\]
Since
\[
\|\nabla_{\mathbf n}\Gamma\|_E^2
=
\sum_{e\in E_{\mathbf n}^+}
c_e\|D_e\Gamma\|_{HS}^{2},
\]
and every edge of type \(\{a,b\}\) satisfies
\[
\|D_e\Gamma\|_{HS}^{2}
=
\Delta_{ab}^{2},
\]
we obtain
\[
\mathcal R_{\mathbf n}
\le
\frac1{\lambda(\mathbf n)}
\sum_{a<b}
w_{ab}(\mathbf n)\Delta_{ab}^{2}.
\]
Averaging over \(\mathbf N\) gives the global bound.
\end{proof}

Thus the Poincar\'e--Dirichlet estimate is the norm-only envelope of the exact
Green-kernel representation. It replaces all nonzero Laplacian modes by the
slowest mode and retains only the edge norms
\[
\|D_e\Gamma\|_{HS}^{2}=\Delta_{\tau(e)}^{2}.
\]
The exact formula is sharper because it preserves both the graph Green kernel
and the Hilbert--Schmidt correlations between contextualized commutator
increments.

\paragraph{Mode-resolved interpretation.}

The spectral representation in
\cref{prop:exact_green_kernel_counts_residual} shows how the order-loss residual
is distributed over the relaxation modes of the adjacent-swap graph. For each
nonconstant eigenmode \(\psi_r^{(\mathbf n)}\), define the commutator flux into
that mode by
\[
\mathcal J_r^{(\mathbf n)}
:=
\sum_{e\in E_{\mathbf n}^+}
c_e\,
\nabla\psi_r^{(\mathbf n)}(e)\,
D_e\Gamma.
\]
Then
\[
\mathcal R_{\mathbf n}
=
\sum_{r\ge1}
\frac{
\|\mathcal J_r^{(\mathbf n)}\|_{HS}^{2}
}{
\lambda_r(\mathbf n)^2
}.
\]
Equivalently, if
\[
\Gamma-\mathbb E_{\mu_{\mathbf n}}\Gamma
=
\sum_{r\ge1}
C_r^{(\mathbf n)}\psi_r^{(\mathbf n)},
\]
then
\[
C_r^{(\mathbf n)}
=
\frac1{\lambda_r(\mathbf n)}
\mathcal J_r^{(\mathbf n)}
\]
and therefore
\[
\mathcal R_{\mathbf n}
=
\sum_{r\ge1}
\|C_r^{(\mathbf n)}\|_{HS}^{2}.
\]
This explains the slack in the Poincar\'e bound. The bound replaces all factors \(1/\lambda_r(\mathbf n)\) by the worst factor \(1/\lambda_1(\mathbf n)\). 

\paragraph{Relation to defect levels.}

For unitary trajectory models, the local swap defect has the explicit
group-commutator expression
\[
\Delta_{ab}^{2}
=
2\left(
d^{2}
-
\left|
\Tr(U_a^\dagger U_b^\dagger U_aU_b)
\right|^{2}
\right),
\]
where \(\Gamma_a(\rho)=U_a\rho U_a^\dagger\). Thus the previous defect-level
interpretation remains useful at the level of the norm-only corollary.

Let
\[
0=\theta_0<\theta_1<\cdots<\theta_K
\]
be the distinct local defect values, and define
\[
\mathcal P_k
:=
\left\{
\{a,b\}:\Delta_{ab}=\theta_k
\right\}.
\]
Then
\[
\sum_{a<b}
w_{ab}(\mathbf N)\Delta_{ab}^{2}
=
\sum_{k=1}^{K}
\theta_k^{2}
\sum_{\{a,b\}\in\mathcal P_k}
w_{ab}(\mathbf N).
\]
This scalar defect-level energy is a useful coarse summary, but it is only a rough envelope. The true order-loss core is determined by the graph Green kernel applied to the contextual commutators
\[
e
\longmapsto
\Gamma_{v_e}\circ
\bigl(
\Gamma_a\circ\Gamma_b-\Gamma_b\circ\Gamma_a
\bigr)
\circ
\Gamma_{u_e}.
\]
Therefore the residual risk is sensitive to the Hilbert--Schmidt inner-product geometry across trajectories.

\section{Observable-adapted residual-risk lemmas}
\label{app:observable_adapted_residual}
\subsection{Observable-adapted and local stratification}
\label{sec:observable_adapted_stratification}

The residual-risk framework above is channel-level and therefore task independent.
This is useful because it gives uniform guarantees for all  downstream observables, but it can also be overly conservative. 
The cost of this universality is that the relevant statistics may have very large support.
For words of fixed length \(m\) over an alphabet \(\mathcal A\) of size \(q=|\mathcal A|\), the full count vector has
\[
\binom{m+q-1}{q-1}
\]
possible sectors. In randomized quantum simulation, \(q\) typically scales with the number of Hamiltonian terms, so full-count stratification can be impractical even when it gives a sharp channel-level guarantee.

In many applications, however, the goal is not to estimate the full channel \(\Gamma(W)\), but rather one or more scalar expectation values. 
Fix an input state \(\rho\) and an observable \(O_R\) supported on a region \(R\). 
For a trajectory \(W\), define the scalar trajectory as the mapping
\[
Y_R(W)
:=
g(W)\Tr[O_R\,\Gamma(W)(\rho)],
\]
where \(g(W)\in\mathbb R\) is the trajectory weight. For ordinary unitary trajectory estimators, \(g(W)=1\); for signed (or QPD-based) estimators, such as TE-PAI, \(g(W)\) includes the signed trajectory weight.

Given a statistic \(S\), define the observable-level residual risk
\[
\mathcal R_R(S)
:=
\mathbb E[\Var(Y_R\mid S)].
\]
Thus \(S\) has zero residual risk for the observable \(O_R\) whenever \(Y_R\) is measurable with respect to \(S\). 

This is weaker than requiring \(S\) to determine the full sampled channel.

Let
\(
\widehat\Phi(W):=g(W)\Gamma(W)
\)
be the weighted channel estimator, and define the linear functional
\[
L_{\rho,O_R}(\Phi):=\Tr[O_R\Phi(\rho)].
\]
Then
\[
Y_R(W)=L_{\rho,O_R}(\widehat\Phi(W)).
\]
Since conditional expectation commutes with fixed linear maps,
\[
\mathbb E[Y_R\mid S]
=
L_{\rho,O_R}\!\left(\mathbb E[\widehat\Phi\mid S]\right).
\]
Therefore the channel-level residual controls the observable-level residual:
\[
\mathcal R_R(S)
\le
\|L_{\rho,O_R}\|^2\,\mathcal R(\widehat\Phi,S).
\]
For the Hilbert--Schmidt norm on superoperators,
\[
\|L_{\rho,O_R}\|
\le
\|O_R\|_F\|\rho\|_F.
\]
The converse need not hold: a statistic can be sufficient, or nearly sufficient, for a particular observable while being far from sufficient for the full channel.

The practical aim of observable-adapted stratification is therefore to use statistics whose support is controlled by the observable locality rather than by the full Hamiltonian. The simplest such choice is a locality-truncated count statistic.

\paragraph{Locality-truncated counts.}

Let
\[
\mathcal A_{R,\ell}\subseteq \mathcal A
\]
be a retained local alphabet for \(O_R\). For example, \(\mathcal A_{R,\ell}\) may consist of all trajectory letters supported in the radius-\(\ell\) neighborhood \(B_\ell(R)\) of the observable support. Define the local count statistic
\[
\mathbf N_{R,\ell}(W)
:=
(N_a(W))_{a\in\mathcal A_{R,\ell}}.
\]
This statistic keeps all count coordinates inside the retained local alphabet and discards all count coordinates outside it. Its support size is
\[
\binom{m+|\mathcal A_{R,\ell}|-1}{|\mathcal A_{R,\ell}|-1},
\]
which can be much smaller than the global count support
\[
\binom{m+|\mathcal A|-1}{|\mathcal A|-1}.
\]
For bounded-range Hamiltonians and fixed-size observables, \(|\mathcal A_{R,\ell}|\) is controlled by the size of the local neighborhood rather than by the total system size.

The residual risk of \(\mathbf N_{R,\ell}\) has two conceptually distinct sources. First, by restricting to a local alphabet, one discards trajectory letters outside the retained region; these may still have an indirect influence on \(O_R\). Second, even inside the retained alphabet, conditioning only on counts forgets the ordering of noncommuting local letters. The following bound separates these two contributions.

\begin{proposition}[Observable residual under local count stratification]
\label{prop:observable_residual_local_counts}
Let \(Y_R\) be the scalar trajectory associated with \(O_R\). Let \(Y_{R,\ell}\) be a local surrogate that depends only on the retained local trajectory over \(\mathcal A_{R,\ell}\), and suppose
\[
\|Y_R-Y_{R,\ell}\|_{L^2}\le \epsilon_{R,\ell}.
\]
Let \(\mathbf N_{R,\ell}\) be the local count statistic on
\(\mathcal A_{R,\ell}\). Then
\[
\mathcal R_R(\mathbf N_{R,\ell})^{1/2}
\le
\epsilon_{R,\ell}
+
\mathcal R(Y_{R,\ell},\mathbf N_{R,\ell})^{1/2}.
\]
Moreover, the local count residual admits the adjacent-swap bound
\[
\mathcal R(Y_{R,\ell},\mathbf N_{R,\ell})
\le
\sum_{\substack{a<b\\a,b\in\mathcal A_{R,\ell}}}
\omega_{ab}^{R,\ell}
\bigl(\delta_{ab}^{R,\ell}\bigr)^2.
\]
Consequently,
\[
\mathcal R_R(\mathbf N_{R,\ell})^{1/2}
\le
\epsilon_{R,\ell}
+
\left[
\sum_{\substack{a<b\\a,b\in\mathcal A_{R,\ell}}}
\omega_{ab}^{R,\ell}
\bigl(\delta_{ab}^{R,\ell}\bigr)^2
\right]^{1/2}.
\]
\end{proposition}
\begin{proof}[Proof of \cref{prop:observable_residual_local_counts}]
For any statistic \(S\), conditional expectation is the \(L^2\)-orthogonal
projection onto the subspace of \(S\)-measurable random variables. Hence
\[
\mathcal R_R(S)^{1/2}
=
\|Y_R-\mathbb E[Y_R\mid S]\|_{L^2}.
\]
Since \(\mathbb E[Y_{R,\ell}\mid S]\) is \(S\)-measurable, the projection
property gives
\[
\mathcal R_R(S)^{1/2}
\le
\|Y_R-\mathbb E[Y_{R,\ell}\mid S]\|_{L^2}.
\]
By the triangle inequality,
\[
\|Y_R-\mathbb E[Y_{R,\ell}\mid S]\|_{L^2}
\le
\|Y_R-Y_{R,\ell}\|_{L^2}
+
\|Y_{R,\ell}-\mathbb E[Y_{R,\ell}\mid S]\|_{L^2}.
\]
Taking \(S=\mathbf N_{R,\ell}\) gives
\[
\mathcal R_R(\mathbf N_{R,\ell})^{1/2}
\le
\epsilon_{R,\ell}
+
\mathcal R(Y_{R,\ell},\mathbf N_{R,\ell})^{1/2}.
\]
The adjacent-swap bound is the scalar version of the weighted
Poincar\'e--Dirichlet argument applied to the retained local trajectory ensemble.
\end{proof}
Here \(\epsilon_{R,\ell}\) is the locality truncation error: it measures the effect of replacing the full scalar trajectory by a local surrogate. The second term is the residual ordering error left after conditioning only on local counts.

The scalar swap defect \(\delta_{ab}^{R,\ell}\) is the worst-case change in the local surrogate caused by swapping one adjacent retained pair \(a,b\):
\[
\delta_{ab}^{R,\ell}
:=
\sup
\left|
Y_{R,\ell}(W)-Y_{R,\ell}(W^{ab})
\right|,
\]
where \(W^{ab}\) denotes the retained local trajectory obtained from \(W\) by swapping one adjacent occurrence of \(a,b\). The weights
\[
\omega_{ab}^{R,\ell}
:=
\mathbb E_{\mathbf N_{R,\ell}}\!\left[
\frac{1}{\lambda_{R,\ell}(\mathbf N_{R,\ell})}
w_{ab}^{R,\ell}(\mathbf N_{R,\ell})
\right]
\]
are the local analogues of the global counts-core weights. Here \(w_{ab}^{R,\ell}(\mathbf n)\) is the total adjacent-swap edge weight exchanging \(a,b\) inside the local count sector \(\mathbf n\), and \(\lambda_{R,\ell}(\mathbf n)\) is the corresponding local graph Laplacian gap.

When an observable-independent estimate inside the local region is preferable, the scalar swap defects can be bounded by channel-level swap defects. 
In particular, for unitary local contexts,
\[
\delta_{ab}^{R,\ell}
\le
\|O_R\|_F\|\rho\|_F\,\Delta_{ab}.
\]
Thus a simpler but more conservative bound is
\[
\mathcal R_R(\mathbf N_{R,\ell})^{1/2}
\le
\epsilon_{R,\ell}
+
\|O_R\|_F\|\rho\|_F
\left[
\sum_{\substack{a<b\\a,b\in\mathcal A_{R,\ell}}}
\omega_{ab}^{R,\ell}
\Delta_{ab}^2
\right]^{1/2}.
\]

\paragraph{Locality error.}

For noncommuting local dynamics terms outside the initial support of \(O_R\) can influence it through operator spreading. 
This effect is captured by locality error \(\epsilon_{R,\ell}\). Lieb--Robinson estimates suggest bounds of the form
\[
\epsilon_{R,\ell}
\lesssim
C\|O_R\|\,|\partial R|\,e^{-\mu(\ell-vt)}.
\]
In our randomized setting this needs to be replaced by a worst-case trajectory analysis. However, the bound in \cref{prop:observable_residual_local_counts} does not depend on a particular Lieb--Robinson estimate; it only requires an \(L^2\) surrogate error. Thus the practical design principle is to choose a local alphabet large enough that \(\epsilon_{R,\ell}\) is acceptable, and then use full local counts to remove the count-level variance inside the retained region, up to the residual local order term.
\paragraph{Example: commuting Ising chain.}

Consider
\[
H
=
w\sum_{j=1}^n Z_j
+
J\sum_{j=1}^n Z_jZ_{j+1}
\]
on a periodic chain. For the local observable \(X_j\), the only Hamiltonian terms that fail to commute with \(X_j\) are
\[
Z_j,
\qquad
Z_{j-1}Z_j,
\qquad
Z_jZ_{j+1}.
\]
Therefore the local count statistic
\[
\mathbf N_j
=
\bigl(
N_{Z_j},
N_{Z_{j-1}Z_j},
N_{Z_jZ_{j+1}}
\bigr)
\]
is exact, so \(
\mathcal R_{X_j}(\mathbf N_j)=0.
\)
It has three coordinates, independent of \(n\), whereas the full count statistic has \(2n\) coordinates.

\paragraph{Exact causal cone statistics.} One could obtain exact observable-level zero residual by conditioning on the full ordered causal cone of \(O_R\), but this statistic may grow rapidly with depth and is not the focus here. Instead, we use simpler locality-truncated count statistics: they are not generally exact, but their support is controlled by the chosen local region and their residual risk decomposes into a locality error and a local order-residual term.
\paragraph{Summary.}
The bounds above motivates a practical stratification strategy. For each target observable \(O_R\), choose a retained local alphabet \(\mathcal A_{R,\ell}\), form the local count statistic \(\mathbf N_{R,\ell}\), and stratify only over those local counts. 
For fixed-size observables and bounded-range Hamiltonians, this replaces a system-size-dependent statistic by one whose dimension is controlled by the chosen locality radius.

The residual risk then has two interpretable contributions, 
\(
\epsilon_{R,\ell},
\)
which measures the influence of omitted nonlocal trajectory information, and
\(
\sqrt{
\sum_{\substack{a<b\\a,b\in\mathcal A_{R,\ell}}}
\omega_{ab}^{R,\ell}\Delta_{ab}^2},
\)
which measures the order information lost by using counts rather than the ordered local trajectory. In the commuting Pauli case, both terms vanish for the visible local alphabet, giving exact zero observable-level residual. In the noncommuting case, the bound suggests a tunable tradeoff: increasing \(\ell\) reduces the locality error, while the remaining local order residual is governed by the commutator structure of the retained terms.

The numerical experiments below test this picture directly by comparing naive sampling, global count stratification, and observable-adapted local count stratification across representative Hamiltonian ensembles.

\begin{proposition}[Exact local stratification in the commuting Pauli case]
Let
\[
H=\sum_{a\in\mathcal A}h_aP_a
\]
be a mutually commuting Pauli Hamiltonian, and let \(O_R\) be a Pauli
observable. Define
\[
\mathcal A_R:=\{a\in\mathcal A:[P_a,O_R]\neq 0\},
\qquad
\mathbf N_R=(N_a)_{a\in\mathcal A_R}.
\]
Then
\[
Y_R(W)=\Tr[O_R\Gamma_W(\rho)]
\]
is measurable with respect to \(\mathbf N_R\). Consequently,
\[
\mathbb E[\Var(Y_R\mid \mathbf N_R)]=0.
\]
\end{proposition}

\begin{proof}
Because all \(P_a\) commute, the trajectory channel is independent of the
letter order:
\[
\Gamma_W=\prod_{a\in\mathcal A}\Gamma_a^{N_a(W)}.
\]
In the Heisenberg picture,
\[
\Gamma_W^\dagger(O_R)
=
\prod_{a\in\mathcal A}
(\Gamma_a^\dagger)^{N_a(W)}(O_R).
\]
If \(a\notin\mathcal A_R\), then \([P_a,O_R]=0\), so
\[
\Gamma_a^\dagger(O_R)=O_R.
\]
Thus all terms outside \(\mathcal A_R\) act trivially on \(O_R\), and
\(\Gamma_W^\dagger(O_R)\) depends only on \(N_a(W)\) for
\(a\in\mathcal A_R\). Hence \(Y_R\) is \(\mathbf N_R\)-measurable.
\end{proof}
\paragraph{Commutator interpretation of the locality error.}
Let \(H=H_{R,\ell}+H_{\mathrm{out}}\), where \(H_{R,\ell}\) contains the retained
local terms and \(H_{\mathrm{out}}\) contains the omitted terms. Define
\[
O_R(t)=e^{itH}O_Re^{-itH},
\qquad
O_R^{(\ell)}(t)=e^{itH_{R,\ell}}O_Re^{-itH_{R,\ell}}.
\]
Duhamel's formula gives
\[
O_R(t)-O_R^{(\ell)}(t)
=
i\int_0^t
e^{isH}
[H_{\mathrm{out}},O_R^{(\ell)}(t-s)]
e^{-isH}
\,ds.
\]
Therefore
\[
\|O_R(t)-O_R^{(\ell)}(t)\|
\le
\int_0^t
\|[H_{\mathrm{out}},O_R^{(\ell)}(t-s)]\|\,ds.
\]
If
\[
H_{\mathrm{out}}
=
\sum_{b\notin\mathcal A_{R,\ell}}h_bP_b,
\]
then
\[
\|O_R(t)-O_R^{(\ell)}(t)\|
\le
\sum_{b\notin\mathcal A_{R,\ell}}
|h_b|
\int_0^t
\|[P_b,O_R^{(\ell)}(s)]\|\,ds.
\]
Thus the locality error measures the integrated commutator influence of omitted
Hamiltonian terms on the locally evolved observable. Lieb--Robinson estimates
upper bound these commutators by a function decaying with the distance from
\(\operatorname{supp}(P_b)\) to \(R\).

\section{Numerical Procedure}
\label{sec:numerical_procedure}

This appendix describes the stratified sampling procedures used in the
event-driven implementation of continuous TE-PAI. Both procedures use finite retained supports without overflow buckets. 
Consequently the resulting estimators are biased, but the bias is controlled explicitly by the omitted probability mass and by any retained strata that receive zero samples after finite-budget rounding. The basic scheme we use is a stratified estimator under proportional allocation. More complex schemes are possible, such as pilot or adaptive schemes.

\subsection{Event-driven sampling of continuous TE-PAI}

Consider a time-dependent Pauli Hamiltonian
\[
    H(s)=\sum_{a\in\mathcal A}c_a(s)P_a,
    \qquad 0\le s\le t,
    \qquad P_a^2=I.
\]
Write
\[
    c_a^+(s):=\max(c_a(s),0),
    \qquad
    c_a^-(s):=\max(-c_a(s),0).
\]
For a fixed TE-PAI angle \(\Delta\), the marked event alphabet is
\[
\mathcal A_{\rm TE}
=
\{(a,+,2),(a,-,2),(a,3):a\in\mathcal A\}.
\]
Recall then that corresponding inhomogeneous event rates are
\begin{align*}
    \kappa_{a,+,2}(s)
    &=
    \frac{2c_a^+(s)}{\sin\Delta},
    \\
    \kappa_{a,-,2}(s)
    &=
    \frac{2c_a^-(s)}{\sin\Delta},
    \\
    \kappa_{a,3}(s)
    &=
    |c_a(s)|\tan(\Delta/2).
\end{align*}
Let
\[
    \mu_\alpha:=\int_0^t \kappa_\alpha(s)\,ds,
    \qquad
    \alpha\in\mathcal A_{\mathrm{TE}},
\]
be the integrated event means. 
The event counts \(N_\alpha\) are independent Poisson random variables with means \(\mu_\alpha\).  
Conditional on the counts, event times for letter \(\alpha\) are sampled independently from density
\[
    f_\alpha(s)=\frac{\kappa_\alpha(s)}{\mu_\alpha},
    \qquad 0\le s\le t,
\]
with the convention that letters with \(\mu_\alpha=0\) produce no events. All events are then merged and time-ordered.

The total type-3 count is
\[
    N_3(W):=\sum_{a\in\mathcal A}N_{a,3}(W),
\]
and the TE-PAI scalar weight magnitude is
\[
    G_\infty
    =
    \exp\!\left(
        2\tan(\Delta/2)
        \int_0^t\sum_{a\in\mathcal A}|c_a(s)|\,ds
    \right).
\]
For an observable \(O\), input state \(\rho\), and sampled trajectory \(W\), the scalar estimator is
\[
    Y(W)
    =
    G_\infty(-1)^{N_3(W)}
    \operatorname{Tr}\!\left[O\,\Gamma_W(\rho)\right].
\]

\subsection{Generic retained-support stratified estimator}
\label{sec:generic_retained_support}

Let \(S(W)\) be a discrete trajectory statistic with known stratum
probabilities \(p_r=\Pr(S=r)\).  Choose a finite retained support
\(K\) and allocate a total budget \(N_{\mathrm{tot}}\) across the retained strata. We use deterministic rounding via Hamilton apportionment. Start from
\[
    N_r^{(0)}=\lfloor N_{\mathrm{tot}}p_r\rfloor,
    \qquad r\in K,
\]
and distribute the remaining samples to the strata with largest remainders
\(N_{\mathrm{tot}}p_r-\lfloor N_{\mathrm{tot}}p_r\rfloor\).  Let
\[
    K_{\mathrm{eff}}:=\{r\in K:N_r>0\}
\]
be the strata that actually receive samples.

For each \(r\in K_{\mathrm{eff}}\), draw
\[
    W_{r,1},\ldots,W_{r,N_r}
\]
independently from the conditional TE-PAI law given \(S(W)=r\), and define
\[
    \widehat\mu_r
    :=
    \frac{1}{N_r}\sum_{q=1}^{N_r}Y(W_{r,q}).
\]
The retained-support estimator is
\[
    \widehat\mu^{\mathrm{bias}}
    :=
    \sum_{r\in K_{\mathrm{eff}}}p_r\widehat\mu_r.
\]
It omits both the truncation tail \(S\notin K\) and retained strata with \(N_r=0\).

Assume the scalar trajectory estimator is uniformly bounded,
\[
    |Y(W)|\le B.
\]
For TE-PAI with \(\rho\) a density matrix, one may take
\[
    B=G_\infty\|O\|_\infty.
\]
Define the omitted probability mass
\[
    p_{\mathrm{omit}}
    :=
    \Pr(S\notin K_{\mathrm{eff}})
    =
    1-\sum_{r\in K_{\mathrm{eff}}}p_r.
\]
Then
\[
    \left|
    \mathbb E[\widehat\mu^{\mathrm{bias}}]-\mu
    \right|
    \le
    B\,p_{\mathrm{omit}}.
\]
If
\[
    p_{\mathrm{miss}}:=\Pr(S\notin K)
\]
is the truncation mass, then
\[
    p_{\mathrm{omit}}
    =
    p_{\mathrm{miss}}
    +
    \sum_{r\in K:\,N_r=0}p_r.
\]
The second term is exactly computable after allocation and vanishes if every retained stratum receives at least one sample.

Conditioned on the allocation,
\[
    \operatorname{Var}(\widehat\mu^{\mathrm{bias}})
    =
    \sum_{r\in K_{\mathrm{eff}}}
    p_r^2\frac{\sigma_r^2}{N_r},
    \qquad
    \sigma_r^2:=\operatorname{Var}(Y\mid S=r).
\]
Consequently,
\[
    \operatorname{MSE}(\widehat\mu^{\mathrm{bias}})
    \le
    \sum_{r\in K_{\mathrm{eff}}}
    p_r^2\frac{\sigma_r^2}{N_r}
    +
    B^2p_{\mathrm{omit}}^2.
\]

\subsection{Local-count statistic with outside type-3 parity}
\label{sec:local_te_pai_sampling_recipe}

We first describe the local statistic used for estimating a single local observable \(O_R\). 
Choose a set of local type-2 count coordinates
\[
    \mathcal I_R\subseteq\mathcal A_{\mathrm{TE}}
\]
adapted to the support and backward light cone of \(O_R\).  For example, in the commuting Ising benchmark with \(O_R=X_j\), we use
\[
    A_2=N^F_{j,2},
    \qquad
    B_2=N^B_{j-1,2},
    \qquad
    C_2=N^B_{j,2}.
\]
More generally, \(\mathcal I_R\) contains the local type-2 event letters whose rotations are visible to the observable.

Let \(\mathcal O_{3}^{\mathrm{out}}\) be the set of outside type-3 event letters whose channel action cancels locally with the TE-PAI sign, and define
\begin{align*}
        N_{3}^{\mathrm{out}}(W)
    &:=
    \sum_{\alpha\in\mathcal O_{3}^{\mathrm{out}}}N_\alpha(W),\\
    \Pi_{3,\mathrm{out}}(W)
    &:=
    N_{3}^{\mathrm{out}}(W)\bmod 2.
\end{align*}
The stratification statistic is
\[
    S_R^{\mathrm{out}}(W)
    =
    \left(
        (N_i(W))_{i\in\mathcal I_R},
        \Pi_{3,\mathrm{out}}(W)
    \right).
\]

For each \(i\in\mathcal I_R\), let \(\mu_i\) be the corresponding Poisson mean, and define
\[
    \mu_{3,\mathrm{out}}
    :=
    \sum_{\alpha\in\mathcal O_{3}^{\mathrm{out}}}\mu_\alpha.
\]
Choose tolerances \(\varepsilon_i>0\) with
\[
    \sum_{i\in\mathcal I_R}\varepsilon_i
    \le
    \varepsilon_{\mathrm{trunc}}.
\]
For each local count coordinate, choose integers \(\ell_i\le u_i\) such that
\[
    \Pr(\ell_i\le N_i\le u_i)
    \ge
    1-\varepsilon_i,
    \qquad
    N_i\sim\operatorname{Poisson}(\mu_i).
\]
The retained count set and retained statistic support are
\[
    K_R
    :=
    \prod_{i\in\mathcal I_R}
    \{\ell_i,\ell_i+1,\ldots,u_i\},
    \qquad
    \widetilde K_R:=K_R\times\{0,1\}.
\]
The omitted truncation mass satisfies
\[
    p_{\mathrm{miss}}
    =
    1-\prod_{i\in\mathcal I_R}
    \Pr(\ell_i\le N_i\le u_i)
    \le
    \sum_{i\in\mathcal I_R}\varepsilon_i
    \le
    \varepsilon_{\mathrm{trunc}}.
\]

For a retained count vector \(s=(n_i)_{i\in\mathcal I_R}\in K_R\),
\[
    p_s^{\mathrm{cnt}}
    =
    \prod_{i\in\mathcal I_R}
    e^{-\mu_i}\frac{\mu_i^{n_i}}{n_i!}.
\]
The outside parity probabilities are
\[
    \Pr(\Pi_{3,\mathrm{out}}=\pi)
    =
    \frac12
    \left(
        1+(-1)^\pi e^{-2\mu_{3,\mathrm{out}}}
    \right),
\]
for $\pi\in\{0,1\}$. Hence the retained stratum weight is
\[
    p_{s,\pi}
    =
    p_s^{\mathrm{cnt}}\,
    \frac12
    \left(
        1+(-1)^\pi e^{-2\mu_{3,\mathrm{out}}}
    \right).
\]

Conditional sampling within a stratum \(r=(s,\pi)\) proceeds as follows.
First fix \(N_i=n_i\) for all \(i\in\mathcal I_R\).  Then sample
    \(K_3^{\mathrm{out}}
    \sim
    \operatorname{Poisson}(\mu_{3,\mathrm{out}})\)
    conditioned on \(K_3^{\mathrm{out}}\equiv \pi \pmod 2\).
    Equivalently,
\[
    \Pr(K_3^{\mathrm{out}}=k\mid k\equiv\pi\!\!\pmod2)
    =
    \frac{
        e^{-\mu_{3,\mathrm{out}}}\mu_{3,\mathrm{out}}^k/k!
    }{
        \frac12(1+(-1)^\pi e^{-2\mu_{3,\mathrm{out}}})
    }.
\]
Given \(K_3^{\mathrm{out}}=k\), distribute the outside type-3 events by the multinomial
\[
    (N_\alpha)_{\alpha\in\mathcal O_3^{\mathrm{out}}}
    \mid K_3^{\mathrm{out}}=k
    \sim
    \operatorname{Multi}
    \left(
        k;
        \left(
            \frac{\mu_\alpha}{\mu_{3,\mathrm{out}}}
        \right)_{\alpha\in\mathcal O_3^{\mathrm{out}}}
    \right).
\]
All other event counts not fixed by the statistic are sampled from their original independent Poisson laws. 
This includes outside non-type-3 events and any local type-3 events not retained in the statistic.  After all counts are sampled, event times are drawn from their conditional densities \(f_\alpha\), merged, and sorted.

The estimator for this statistic is obtained by applying the generic retained-support procedure of Sec.~\ref{sec:generic_retained_support} with
\[
    S=S_R^{\mathrm{out}},
    \qquad
    K=\widetilde K_R,
    \qquad
    B=G_\infty\|O_R\|_\infty.
\]
Thus the bias is bounded by
\[
    \left|
    \mathbb E[\widehat\mu_R^{\mathrm{bias}}]-\mu_R
    \right|
    \le
    G_\infty\|O_R\|_\infty
    \left[
        \varepsilon_{\mathrm{trunc}}
        +
        \sum_{r\in\widetilde K_R:\,N_r=0}p_r
    \right].
\]

\subsection{Stratification by the total number of \(\pi\)-flip events}
\label{sec:te_pai_n3_stratification}

The second statistic is the total number of type-3, or \(\pi\)-flip, events:
\[
    S(W):=N_3(W)=\sum_{a\in\mathcal A}N_{a,3}(W).
\]
This is a one-dimensional statistic. 
It is much coarser than local or full marked-count statistics because it does not record which type-3 events occurred, where they occurred in time, or how many type-2 events occurred. 
Its advantage is that it captures the global TE-PAI sign exactly:
\[
    (-1)^{N_3(W)}
\]
is constant within each stratum.

Let
\[
    \mu_3
    :=
    \sum_{a\in\mathcal A}\mu_{a,3}
    =
    \tan(\Delta/2)
    \int_0^t\sum_{a\in\mathcal A}|c_a(s)|\,ds.
\]
Then
\[
    N_3\sim\operatorname{Poisson}(\mu_3).
\]
Choose integers \(\ell\le u\) such that
\[
    \Pr(\ell\le N_3\le u)
    \ge
    1-\varepsilon_{\mathrm{trunc}},
\]
and define
\[
    K_3:=\{\ell,\ell+1,\ldots,u\}.
\]
For \(k\in K_3\),
\[
    p_k
    =
    \Pr(N_3=k)
    =
    e^{-\mu_3}\frac{\mu_3^k}{k!},
\]
and
\[
    p_{\mathrm{miss}}
    =
    1-\sum_{k=\ell}^{u}p_k
    \le
    \varepsilon_{\mathrm{trunc}}.
\]

To sample conditionally on \(N_3=k\), first draw the \(k\) type-3 events from the normalized type-3 intensity. Equivalently, draw each type-3 label \(a\) with probability
\[
    \frac{\mu_{a,3}}{\mu_3},
\]
and then draw its time from
\[
    f_{a,3}(s)
    =
    \frac{\kappa_{a,3}(s)}{\mu_{a,3}}.
\]
All type-2 events are sampled independently from their original inhomogeneous Poisson processes.
The full event list is then merged and sorted.  Within the stratum \(N_3=k\), the scalar estimator is
\[
    Y(W)
    =
    G_\infty(-1)^k
    \operatorname{Tr}\!\left[O\,\Gamma_W(\rho)\right].
\]

The estimator is again the generic retained-support estimator of Sec.~\ref{sec:generic_retained_support}, now with \(S=N_3\), \(K=K_3\), and \(B=G_\infty\|O\|_\infty\).  Hence
\[
    \left|
    \mathbb E[\widehat\mu^{\mathrm{bias}}]-\mu
    \right|
    \le
    G_\infty\|O\|_\infty
    \left[
        \varepsilon_{\mathrm{trunc}}
        +
        \sum_{k\in K_3:\,N_k=0}p_k
    \right],
\]
and
\begin{align*}
    \operatorname{MSE}(\widehat\mu^{\mathrm{bias}})
    &\le
    \sum_{k\in K_{\mathrm{eff}}}
    p_k^2\frac{\sigma_k^2}{N_k} \\
    &+
    G_\infty^2\|O\|_\infty^2
    \left[
        \varepsilon_{\mathrm{trunc}}
        +
        \sum_{k\in K_3:\,N_k=0}p_k
    \right]^2,
\end{align*}
where \(\sigma_k^2=\operatorname{Var}(Y\mid N_3=k)\).

This statistic has minimal bookkeeping overhead: its support size is \(|K_3|=u-\ell+1\) and depends on the Hamiltonian only through the total type-3 mean \(\mu_3\). 
Its expected variance reduction is therefore limited, but it is cheap and directly resolves the dominant global quasiprobability sign structure.

\end{document}

%% file: header.tex
\usepackage{graphicx}  % Include figure files
\usepackage{dcolumn}   % Align table columns on decimal point
\usepackage{bm}        % Bold math
\usepackage[colorlinks=true, allcolors=blue]{hyperref} % clickable links
\usepackage{xcolor}    % optional, for colored text or highlights
\usepackage{amsthm, amssymb, amsmath}
\usepackage{mathrsfs}
\usepackage{enumitem}
\usepackage{float}
\usepackage[capitalize]{cleveref}
\usepackage[linesnumbered,algoruled,boxed,lined]{algorithm2e}
\usepackage{mdframed}

\newtheorem{theorem}{Theorem}

\newtheorem{lemma}{Lemma}

\newtheorem{corollary}{Corollary}

\newtheorem{definition}{Definition}

\newcommand{\ket}[1]{| #1 \rangle}

\newtheorem{proposition}{Proposition}

\DeclareMathOperator{\Tr}{Tr}
\DeclareMathOperator{\Var}{Var}